\documentclass[11pt, letter]{article}

\usepackage{fullpage}

\title{The Price of Anarchy in One-Sided Allocation Problems\\ with Multi-Unit Demand Agents}
\author{
  Sissi Jiang\\
  \texttt{yuhe.jiang@mail.mcgill.ca}
  \and
  Ndiame Ndiaye,\\
  \texttt{ndiame.ndiaye@mail.mcgill.ca}
    \and
  Adrian Vetta\\
  \texttt{adrian.vetta@mcgill.ca}
  \and
  Eggie Wu\\
  \texttt{qihan.wu@mail.mcgill.ca}
}
\usepackage{amsthm}
\usepackage{thmtools}

\newtheorem{theorem}{Theorem}[section]
\newtheorem{example}{Example}

\newtheorem{lemma}[theorem]{Lemma}
\newtheorem{corollary}[theorem]{Corollary}

\theoremstyle{remark}

\usepackage{stmaryrd}
\usepackage{mathrsfs}
\usepackage{amsmath}
\usepackage{bbm}
\usepackage{graphicx}
\usepackage{subcaption}
\usepackage{amsfonts}
\usepackage[normalem]{ulem}

\title{The Price of Anarchy of the\\ Asymmetric One-Sided Allocation Problem}
\author{
  Sissi Jiang\\
  \texttt{yuhe.jiang@mail.mcgill.ca}
  \and
  Ndiame Ndiaye,\\
  \texttt{ndiame.ndiaye@mail.mcgill.ca}
    \and
  Adrian Vetta\\
  \texttt{adrian.vetta@mcgill.ca}
  \and
  Eggie Wu\\
  \texttt{qihan.wu@mail.mcgill.ca}
}

\newcommand{\genTime}[2]{t_{#1}(#2)} 
\newcommand{\genValue}[1]{(#1,v_{-i})} 

\newcommand{\NETime}[1]{t_{#1}} 
\newcommand{\SeqTime}[1]{\hat{t}_{#1}} 
\newcommand{\SingTime}[1]{\tilde{t}_{#1}} 
\newcommand{\IndTime}[2]{t^{#1}_{#2}}

\newcommand{\NEValue}[1][]{v_{#1}} 
\newcommand{\TrueValue}[1][]{v'_{#1}} 

\newcommand{\EpsValue}[1]{\hat{u}^{\varepsilon}_{#1}} 
\newcommand{\SeqValue}[1]{\hat{u}_{#1}} 
\newcommand{\IndValue}[2][]{u^{#2}_{#1}}

\newcommand{\remSet}[2]{r^{#1}(#2)}
\newcommand{\remAmm}[3]{q_{#1}^{#2}(#3)}
\newcommand{\conRate}[4]{c_{#1,#2}^{#3}(#4)}

\newcommand{\valInd}[0]{\mathcal{U}}
\newcommand{\valGroup}[0]{\mathcal{V}}

\newcommand{\sumInd}[1]{C^{Y,#1}_i(t)}

\newcommand{\fn}[2]{#1\left(#2\right)}

\newcommand{\set}[1]{\left\{#1\right\}}
\newcommand{\floor}[1]{\left\lfloor#1\right\rfloor}
\newcommand{\ceil}[1]{\left\lceil#1\right\rceil}

\newcommand{\abs}[1]{\left| #1 \right|}

\newcommand{\opInt}[2]{\left(#1,#2\right)}
\newcommand{\clInt}[2]{\left[#1,#2\right]}

\newcommand{\bigO}[1]{O\left(#1\right)}

\newcommand{\bb}[0]{\mathbb}

\begin{document}

\maketitle
    
\begin{abstract}
    We study ``fair mechanisms" for the (asymmetric) one-sided allocation problem with $m$~items and $n$~multi-unit demand agents with additive, unit-sum valuations.
    The symmetric case ($m=n$), the one-sided matching problem, has been studied extensively
    for the special class of unit demand agents, in particular with respect to the folklore {\em Random Priority} mechanism and the {\em Probabilistic Serial} mechanism, introduced by Bogomolnaia and Moulin~\cite{BM01}. These are both fair mechanisms and attention has focused on their structural properties, incentives, and performance with respect to social welfare.
    Under the standard assumption of unit-sum valuation functions,
    Christodoulou et al.~\cite{CFF15} proved that the price of anarchy is $\Theta(\sqrt{n})$ in the one-sided matching problem for both the Random Priority and Probabilistic Serial mechanisms.
    Whilst both Random Priority and Probabilistic Serial are {\em ordinal mechanisms}, these approximation guarantees are the best possible even for the broader 
    class of {\em cardinal mechanisms}.
    
    To extend these results to the general setting of the one-sided allocation problems there are two technical obstacles.
    One, asymmetry ($m\neq n$) is problematic especially when the number of items is much greater than the number of 
    items, $m\gg n$. Two, it is necessary to study multi-unit demand agents rather than simply unit demand agents.
    Our approach is to study a natural cardinal mechanism variant of Probabilistic Serial, which we call {\em Cardinal Probabilistic Serial}. We present structural theorems for this mechanism and use them to 
    obtain bounds on the price of anarchy. Our first main result is an upper bound 
    of $O(\sqrt{n}\cdot \log m)$ on the price of anarchy for the asymmetric one-sided allocation problem with multi-unit demand agents. This upper bound applies to both Probabilistic Serial and Cardinal Probabilistic Serial and there is
    a complementary lower bound of $\Omega(\sqrt{n})$ for any fair mechanism. That lower bound is
    unsurprising. More intriguing is our second main result: the price of anarchy degrades with the number of items.
    Specifically, a logarithmic dependence on the number of items is necessary as we show a lower bound
    of $\Omega(\min \{n\, , \, \log m\})$ for both Probabilistic Serial and Cardinal Probabilistic Serial. 
\end{abstract}

\section{Introduction}

In the {\em one-sided matching problem} a set of $m$ items must be matched in a {\em fair} manner to a set of $n=m$ (symmetry) agents. This is a classical problem in economics and computer science with numerous practical applications, such as assigning 
children to schools, patients to doctors, workers to tasks, social housing to people, etc.
Consequently, there has been a huge amount of research concerning matching mechanisms, their incentive and structural properties, and the social quality of the outcomes they induce. Of course, these mechanisms are restricted by the fact that the allocation must be a matching. Equivalently, this constraint can be viewed as an assumption of {\em unit demand} valuation functions, where each agent desires at most one good. 
However, unit demand valuations are very restrictive. Indeed, in mechanism design primary focus is on {\em multi-unit demand} valuations and Budish et al.~\cite{BCKM13} highlight the importance of moving beyond unit-demand agents in the field of {\em fair} mechanism design. Moreover, in many practical applications the number of items differs from the number of
agents ($m\neq n$, asymmetry) and/or the agents have {\em multi-unit demand} valuations. For example, in estate division or the allocation of shifts to employees, university courses to students, landing and hanger slots to airlines, etc.
This motivates our work: we study fair allocation mechanisms
for the asymmetric one-sided allocation problem with multi-unit demand agents
and analyse the quality of the outcomes they produce with respect to social welfare.

\subsection{Background}
The one-sided matching problem with indivisible items 
was formally introduced by Hylland and Zeckhauser~\cite{HZ79} in 1979, where they
studied the competitive equilibrium from equal incomes (CEEI) mechanism.
This mechanism is ``fair" by the equal incomes assumption. It is also {\em envy-free} but not strategy-proof and, indeed, early work in the economics community focused on the structural and incentive properties of matching mechanisms. For example, Zhou~\cite{Z90} gave an impossibility result showing the non-existence
of a mechanism that is simultaneously strategy-proof, pareto optimal, and symmetric.
See~\cite{AS13,SU11} for surveys on the one-sided matching problem and on matching markets
more generally.

Since monetary transfers are typically not allowed in the
one-sided matching problem, it belongs to the field of
mechanism design without money~\cite{PT09}. 
A folklore mechanism in this realm is Random Priority (RP). 
Applied to the one-sided matching problem, this mechanism orders the agents
uniformly at random. The agents then, in turn, select their favorite item that has not 
previously been selected. 
This mechanism, also popularly known as 
Random Serial Dictatorship (RSD)~\cite{AS98}, is strategy-proof. 

Another prominent mechanism is Probabilistic Serial (PS), introduced by Bogomolnaia and Moulin~\cite{BM01} in 2001. This is a ``consumption'' mechanism: to begin, every agent {\em consumes} their favorite item at the same {\em consumption rate}. When the favorite item of an agent is completely consumed
(that is, together all the agents have consumed exactly one unit of that item)
then this agent switches to consume its next favorite item, etc.
Since its discovery, Probabilistic Serial has become the most well-studied mechanism for the one-sided matching problem.
It has many desirable properties such as envy-freeness and ordinal efficiency when the agents are truthful~\cite{BM01}. 
However, unlike Random Priority, it is not strategy-proof and some of its desirable properties fail to hold when the agents are strategic~\cite{EK16}.
Several extensions to the mechanism have been proposed; see, for example,
\cite{KS06,BCKM13,ASS20}. Aziz et al. studied the manipulability of Probabilistic Serial~\cite{AGM15b} and proved the existence of pure strategy Nash equilibria under the mechanism~\cite{AGM15a}. 

An important recent line of research in the computer science community has been to quantify the 
social welfare of allocations induced by a mechanism in comparison to the optimal obtainable social welfare.
Two approaches abound in the literature~\cite{FFZ14,Z22}. First is the {\em approximation ratio}, where agents are assumed to report truthfully to the mechanism. Second, and more interestingly from a game-theoretic perspective, is the {\em price of anarchy}, where agents are assumed to be strategic~\cite{CFF15}. However, for mechanism design without money, these measures are of little interest without a normalization assumption. As a result, the standard normalization assumption~\cite{BCH12,CKK12,CFF15,FFZ14,Z22} is that the valuation function of each 
agent is {\em unit-sum}. Specifically, agent $i$ has a non-negative value $v'_i(j)$ for item $j$ and $\sum_{j} v'_i(j)=1$.
Under the unit-sum assumption, a breakthrough result of Christodoulou et al.~\cite{CFF15} is that 
price of anarchy is $\Theta(\sqrt{n})$ for both the Random Priority and Probabilistic Serial
mechanisms for the one-sided matching problem.

We remark that both Random Priority and Probabilistic Serial have the characteristic that they are
{\em ordinal mechanisms}. Specifically, rather than requiring the entire valuation function
of each agent, they need only the preference ordering on the items induced by the valuation function.
Interestingly, despite being ordinal mechanisms, these bounds are the best possible even
for the broader class of {\em cardinal mechanisms} where agents are required to submit their
entire valuation function~\cite{CFF15}.

From a practical perspective, unit-sum valuation constraints have wide application. One topical example is sports drafts, such as salary caps in US major sports leagues or the Indian Premier League auction in cricket. Here a unit-sum restriction on bids or salaries imposes equity across teams. Moreover, cardinal mechanisms are preferable as they allow teams to specialize their strategies, such as bidding aggressively for a few specialist positions or bidding conservatively across many positions to build team strength. 

\subsection{Overview and Results}
The aim of this work is to extend the study of one-sided allocation problems
beyond matchings to general allocations. Ergo, we consider asymmetric allocation problems
and allow for agents with multi-unit demand valuation functions rather than
unit demand valuations. 
In particular, we desire a mechanism with provably good social welfare guarantees.
To do this, we primarily focus on cardinal mechanisms rather than ordinal mechanisms. Specifically, we design a cardinal variant of the Probabilistic Serial
(recall that Probabilistic Serial is an ordinal mechanism): at any point in time each agent simultaneously  consumes {\em multiple} items, with the consumption rates of the items weighted in accordance with the cardinal valuation function of the agent.
We call this the {\em Cardinal Probabilistic Serial} (CPS) mechanism and define it formally in Section~\ref{sec:CPS} along with examples.

In Section~\ref{sec:technical} we present structural theorems for the Cardinal Probabilistic Serial mechanism. 
In Section~\ref{sec:POA} we prove our main results.
First, we use our structural theorems to show that
the Cardinal Probabilistic Serial mechanism has a price of anarchy of
$O(\sqrt{n}\cdot \log m)$ for the asymmetric one-sided allocation problem and multi-unit demand
agents with additive unit-sum valuations.
The methodology we develop also applies to the ordinal Probabilistic Serial mechanism giving it the same price of anarchy bound for
multi-unit demand agents.\footnote{We remark that the proof of~\cite{CFF15} for Probabilistic Serial does not apply with multi-unit demand agents, even in the simple symmetric ($m=n$) setting.}
Second, we prove a lower bound of 
$\Omega(\sqrt{n})$ on the price of anarchy for any ``fair" mechanism, where a mechanism is deemed fair if each agent obtains the same number of items in expectation (as is the case for RP, PS and CPS). 
Third, we present a more intriguing lower bound: the price of anarchy degrades with the number of items.
Specifically, a lower bound of $\Omega(\min \{n\, , \, \log m\})$ for both Probabilistic Serial 
and Cardinal Probabilistic Serial is shown; thus, a logarithmic dependence on the number of items is necessary for the price of anarchy in the asymmetric one-sided allocation problem.

Finally, we wrap up in Section~\ref{sec:related} with dicussions on (i) the price of stability, (ii) 
the relative merits {\em in practice} of Probabilistic Serial and Cardinal Probabilistic Serial, and (iii) 
the performance of (extensions of) the Random Priority mechanism in the asymmetric one-sided allocation problem.

\section{The One-Sided Allocation Problem}
In this section we present the asymmetric one-sided allocation problem with multi-unit demand agents. 
There is a set $I$ of $n$ agents and a set $J$ of $m$ items.
Each agent $i\in I$ has a non-negative value $v'_i(j)$ for item $j\in J$.
The agents have additive {\em multi-unit demands}: agent $i$ has a value $v'_i(S)=\sum_{j\in S} v'_i(j)$ for any collection $S\subseteq J$ of items.\footnote{In contrast, for a {\em unit demand} agent $i$, we have $v'_i(S)=\max_{j\in S} v'_i(j)$.}
Furthermore, we assume that valuation functions are {\em unit-sum}, that is,
$\sum_{j\in J} v'_i(j) =1$ for every agent $i$.
Denote by~$\valInd$ the set of unit-sum valuation functions.

Our focus is on direct revelation mechanisms. Given a unit-sum valuation function $v'_i$, agent $i$ can report to the allocation mechanism $M$ a, possibly non-truthful, unit-sum valuation function $v_i$.
We denote the space of feasible reports the mechanism may receive by $\valGroup\equiv \valInd^n$.
Given a set $v=\{v_1,v_2,\dots v_n\}\in \valGroup$ of reported valuations, let
$v_{-i}=\{v_1, \dots, v_{i-1}, v_{i+1},\dots v_n\}$ be the set of reported valuations excluding agent $i$.
We define $M(v) = M(\{v_i, v_{-i}\})$ to be the bundle of items allocated to agent $i$ by the mechanism $M$ given the reported valuations $v=\{v_i, v_{-i}\}$. 
We say that $v'_i(M(v))$ is the {\em payoff} to agent $i$, the true value they have for this bundle.
Further, $v_i$ is a {\em best response} to $v_{-i}$ if it maximizes the resultant 
payoff to agent $i$, that is,  $v_i= \text{argmax}_{\hat{v}_i\in \valInd} v'_i( M(\{\hat{v}_i, v_{-i}\}))$. The reported valuation $v$ is a {\em Nash equilibrium} if $v_i$ is 
best response to $v_{-i}$, for every agent $i\in I$.
Denote by $NE(v')$ the set of valuations which are Nash equilibria with respect to the true valuations $v'=\{v'_1, v'_2,\dots, v'_n\}$.

The {\em social welfare} of the allocation given by the mechanism
is $\sum_{i\in I} v'_i(M(v))$.
Observe that, for additive valuation functions, the social welfare is maximized by simply assigning each item to
the agent that values it the most. Thus the optimal welfare is
$OPT(v')= \sum_{j\in J} \, \max_{i\in I} v_i(j)$.
The {\em price of anarchy} is the worst-case ratio between the optimal welfare
and the social welfare of the {\em worst} Nash equilibrium, namely 
$\sup_{v'}\, \sup_{v\in NE(v')} \, \frac{OPT(v')}{\sum_{i\in I} v'_i(M(v))}$.
Similarly, the {\em price of stability} is the worst-case ratio between the optimal welfare and the social welfare of the {\em best} Nash equilibrium.

\section{The Cardinal Probabilistic Serial Mechanism}\label{sec:CPS}
We are now ready to present our allocation mechanism. We generalize the ordinal mechanism Probabilistic Serial to a cardinal mechanism.
In this consumption mechanism, {\em Cardinal Probabilistic Serial} (CPS), at any time the agents simultaneously consume multiple items rather than just their most preferred remaining item.\footnote{This idea is analogous to the spreading of bids over items in the CEEI mechanism and in trading post games.}
Specifically, at any time, the total consumption rate (speed) of an agent over all items is one but this consumption rate is split amongst the remaining items in proportion to their value to the agent.
Let's now formalize the mechanism.

\subsection{A Cardinal Variant of Probabilistic Serial}

Let $v=\{v_1,v_2,\dots v_n\}\in \valGroup$ be the reported unit-sum valuations.
Each item has a size (quantity) of one unit, and each
agent $i$ has a consumption rate of $1$ at any time. 
At time $t\in\clInt{0}{\frac{m}{n}}$,\footnote{At time $\frac{m}{n}$ all the items have been consumed because there are $m$ units to consume and each of the $n$ agents each consume at a rate of 1.} let $\remSet{v}{t}$ be the set of {\em remaining} items, that is, the items that have not yet been entirely consumed.
Each agent partitions its consumption over the remaining items in proportion to their values. Specifically, at time $t$, the consumption rate of an item $j\in\remSet{v}{t}$ by agent $i$ is denoted  $\conRate{i}{j}{v}{t}$. Formally, if $j\in \remSet{v}{t}$ and $\exists \ell \in \remSet{v}{t}$ such that $v_{i}(\ell)>0$ 
then: 
\begin{equation*}
\conRate{i}{j}{v}{t} 
\, =\, \frac{v_{i}(j)}{\sum\limits_{\ell \in \remSet{v}{\tau}} v_{i}(\ell)}.
\end{equation*}

If $j$ has already been entirely consumed by time $t$, that is, $j\notin \remSet{v}{t}$, then $\conRate{i}{j}{v}{t}=0$. However, if the agent has no value for any of the remaining items then any partition of the consumption rate over the remaining items is allowed and is consistent. 
For our analysis, if this situation ever arises then we assume the consumption rates are chosen adversarially. 

Let $q\colon [m]\times \clInt{0}{\frac{m}{n}}\times \valGroup\rightarrow [0,1]$ be the function denoting the quantity of an item $j$ remaining at time $t$ given the strategies $v$.
Thus:
\begin{equation*}
\remAmm{j}{v}{t}
\, =\, 1-\int_{\tau=0}^t \left(\sum_{i\in I} \conRate{i}{j}{v}{\tau}\right)d\tau
\end{equation*} 

Observe
$j\in \remSet{v}{t}$ if and only if $\remAmm{j}{v}{t}>0$; that is, item $j$ is available at time $t$ if and only if a positive quantity of the good remains.
In particular, at time $0$ we have $\remAmm{j}{v}{0}=1$ and $\remSet{v}{0}=J=[m]=\{1,2,\dots,m\}$. Whilst at time $\frac{m}{n}$ we have $\remAmm{j}{v}{t}=0$ and $\remSet{v}{t}=\emptyset$.
We say that the {\em consumption time} of item $j$ is the earliest time $t$ at which $\remAmm{j}{v}{t}=0$.

We then allocate the items to the agents as follows.
Let $\gamma_{i,j}$ to be the total amount agent $i$ consumed of item $j$.
Then item $j$ is randomly allocated to agent $i$ with probability $\gamma_{i,j}$.

We remark that Cardinal Probabilistic Serial does generalize Probabilistic Serial. Specifically, 
Lemma~\ref{lem:epsilon} 
in the Appendix formally shows how this cardinal mechanism can simulate 
the ordinal mechanism Probabilistic Serial.

\subsection{Examples}
The Cardinal Probabilistic Serial mechanism is easy to understand with some examples. 

\begin{example}
First, consider the valuations in Figure~\ref{fig:valuations},
for three agents (A, B, C) and three items (1, 2, 3).

\begin{figure}[!ht]
  \centering
  \includegraphics[width=0.2\linewidth]{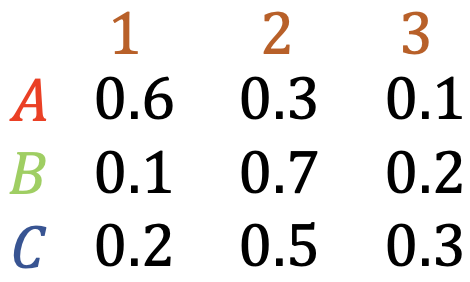}
  \caption{Reported valuations for the three agents.}
  \label{fig:valuations}
\end{figure}

At time $t=0$ the agents consume the items in proportion to their valuations. For example, agent A has a consumption rate of $0.6$ for the item 1, 
$0.3$ for the item 2 and $0.1$ for the item 3.
An important observation is that the consumption rates depend only on the set of remaining items $\remSet{v}{t}$. In particular, the consumption rates remain constant until the
next item has been entirely consumed.
This is at time $t_1=\frac23$, the consumption time of item 2, because the total consumption rate of item 2 until this time is $0.3+0.7+0.5=1.5$. The situation at this time is illustrated in Figure~\ref{fig:t1}. Here agents $A, B, C$ have the colours red, green and blue, respectively, and height in the bar chart represents the amount of an item consumed.

\begin{figure}[!ht]
  \centering
  \includegraphics[width=0.8\linewidth]{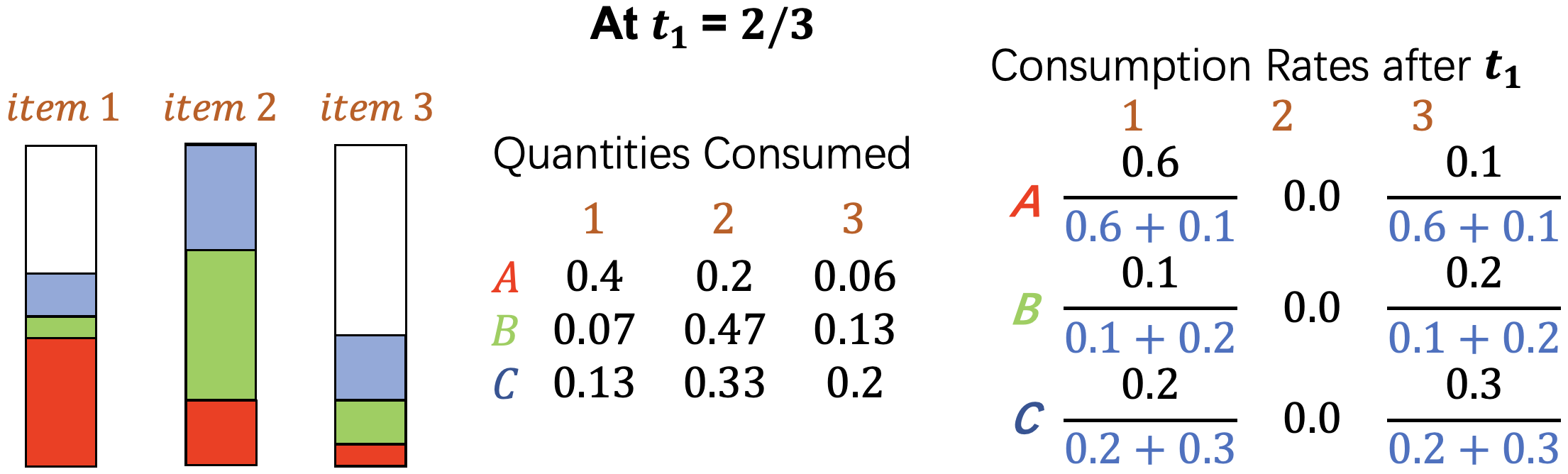}
  \caption{The situation at the consumption time of item 2.}
  \label{fig:t1}
\end{figure}

Because item 2 is no longer available after $t_1=\frac23$, the consumption rates are now updated. 

For example, the consumption rates of agent B for items 1 and 3 are $\frac13$ and $\frac23$, respectively, because its valuations for these 
items are $0.1$ and $0.2$, respectively.
These consumption rates are constant until the consumption time of item~1 at $t_2 \approx 0.92$. 
At this time the amount each agent has consumed is
illustrated in Figure~\ref{fig:t2}. 

\begin{figure}[!ht]
  \centering
  \includegraphics[width=0.8\linewidth]{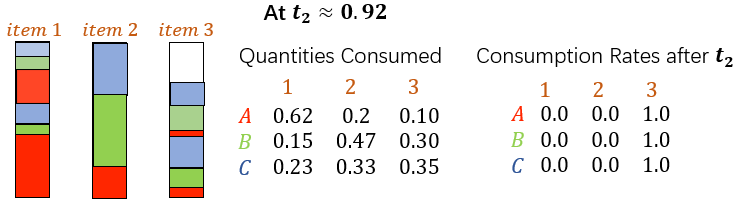}
  \caption{The situation at the consumption time of item 1.}
  \label{fig:t2}
\end{figure}

Now only item 3 remains so each agent consumes it at rate $1$. The consumption time of this last item is $t=1$ and the algorithm terminates.

\begin{figure}[!ht]
  \centering
  \includegraphics[width=0.5\linewidth]{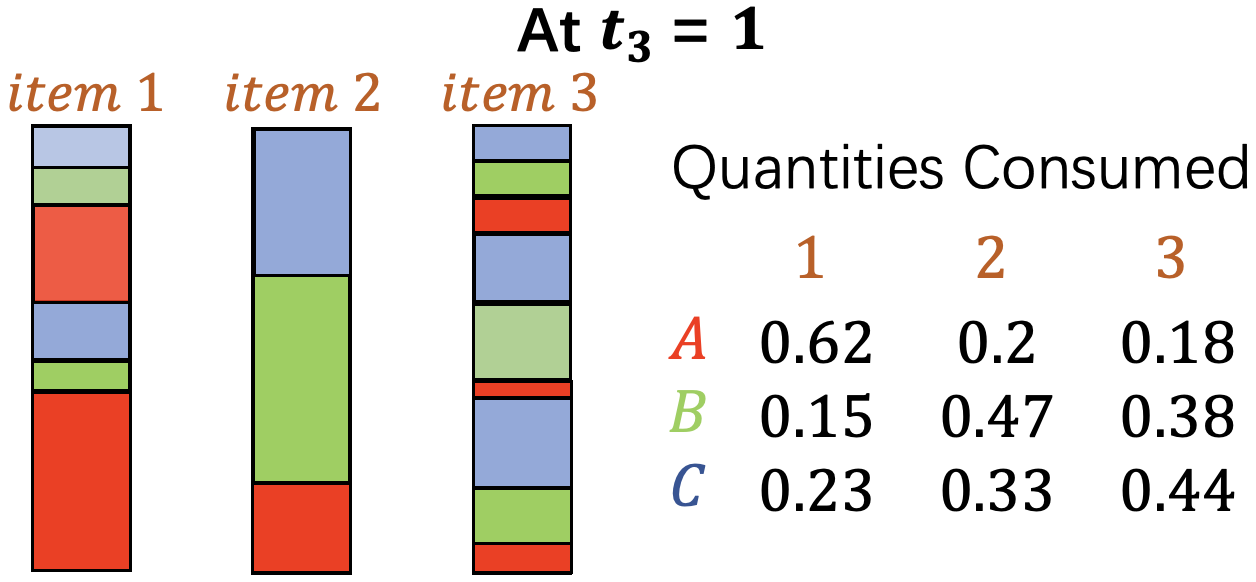}
  \caption{The situation at the consumption time of item 3.}
  \label{fig:t3}
\end{figure}

At this point, see Figure~\ref{fig:t3}, the agent A has consumed the item with quantities $(0.62,0.2,0.18)$, respectively. The agent B has consumed the item with quantities $(0.15,0.47,0.38)$. And agent C has consumed the item with quantities $(0.23, 0.33, 0.44)$.
Thus, item $1$ is assigned to agents A, B and C with probabilities $0.62, 0.15$ and $0.23$, respectively, etc.    
\end{example}

\begin{example}\label{ex:PO}
Let there be two agents (A and B) and two items (1 and 2).
Let $v'_A=(\frac23, \frac13)$ and $v'_B=(\frac13, \frac23)$.
Thus agent A prefers item 1 and the agent B prefers item 2. If the agents report truthfully $v=v'$
then agent A obtains the bundle $(\frac23, \frac13)$
for a payoff of $(\frac23)^2 + (\frac13)^2=\frac59$. Similarly the 
agent B obtains the bundle $(\frac13, \frac23)$
for a payoff of $(\frac13)^2 +(\frac23)^2 =\frac59$. But this is not an equilibrium. In particular, if the agent A deviates and reports $(1,0)$
then it will obtain the bundle $(\frac34,\frac14)$ for a payoff of $\frac34\cdot \frac23+ \frac14\cdot \frac13=\frac{7}{12}> \frac59$.
    
\end{example}

\subsection{The Social Welfare of Equilibria}
Example~\ref{ex:PO} shows that Cardinal Probabilistic Serial, like Probabilistic Serial, is not strategy-proof and motivates studying strategic agents 
and Nash equilibria under this mechanism.
We are especially interested in calculating the price of anarchy of
the mechanism.
For the ordinal mechanism Probabilistic Serial the price of anarchy is
known for the one-sided matching problem due to the work Christodoulou et al~\cite{CFF15}. 
\begin{theorem}\cite{CFF15}\label{PS-Anarchy}
For the one-sided matching problem with unit-sum valuations, the price of anarchy of Probabilistic Serial is $\Theta(\sqrt{n})$.
\end{theorem}
In fact, this guarantee extends beyond Nash equilibria to coarse correlated equilibria and to Bayesian settings.
Furthermore, they show this guarantee is the best possible.
\begin{theorem}\cite{CFF15}\label{Anarchy-Gen}
For the one-sided matching problem with unit-sum valuations, the price of anarchy of any unit-sum mechanism is $\Omega(\sqrt{n})$.
\end{theorem}
The aim of this paper is to extend to results of~\cite{CFF15} to the
asymmetric one-sided allocation problem with multi-unit demand valuations using the Cardinal Probabilistic Serial mechanism.
This we achieve in Section~\ref{sec:UB} with our main positive result:

\begin{restatable}{theorem}{CPSAnarchy}\label{thm:CPS-Anarchy}
    For the asymmetric one-sided allocation problem with multi-unit demand agents, the price of anarchy 
    of Cardinal Probabilistic Serial is $O(\sqrt{n}\cdot \log m)$.
\end{restatable}

Moreover, we then show in Section~\ref{sec:UB-PS} that this upper bound on the price of anarchy also applies the the standard Probabilistic Serial mechanism. 

\begin{restatable}{theorem}{PSAnarchy}\label{thm:PS-Anarchy}
    For the asymmetric one-sided allocation problem with multi-unit demand agents, the price of anarchy 
    of Probabilistic Serial is $O(\sqrt{n}\cdot \log m)$.
\end{restatable}

These results are tight to within the logarithmic factor. In particular, in Section~\ref{sec:LB} we show that a lower bound of $\Omega(\sqrt{n})$ holds in this setting for {\em any} fair mechanism.
Furthermore, our main negative result is that a logarithmic dependence on $m$ is necessary in the lower bound, 
for both Probabilistic Serial and Cardinal Probabilistic Serial. In particular, the price of anarchy must degrade 
with the number of items. 

\begin{restatable}{theorem}{CPSAnarchyAsym}\label{thm:CPS-Anarchy-m}
    For the asymmetric one-sided allocation problem with multi-unit demand agents, the price of anarchy
    of both Probabilistic Serial and Cardinal Probabilistic Serial is at least $\Omega(\min \{n\, ,\, \log m\})$.
\end{restatable}

The rest of the paper is dedicated to proving these results.

\section{Single-Minded and Sequential Bidding}\label{sec:technical}

To quantify the price of anarchy we require an understanding of the allocations and payoffs induced at a Nash equilibrium $v\in \valGroup$.
This is difficult to do directly. So a standard approach is, for each agent $i$, to fix the strategies $v_{-i}$ of the other agents and hypothesize
about the payoff obtainable if the agent plays an alternative strategy.
This lower bounds the payoff obtained by the strategy $v_i$ because it is a
best response to $v_{-i}$. Summing over all agents then gives a lower bound on social welfare.

But what alternative strategies should be considered? In Sections~\ref{sec:single-minded} and \ref{sec:sequential}, we study two simple strategies for each agent: single-minded bidding and sequential bidding.
We prove structural properties of these strategies 
and then use these properties to prove a technical lemma.
This technical lemma can be viewed as a generalization 
to the asymmetric allocation problem with multi-unit demand agents of the main technical lemma of~\cite{CFF15}
for the symmetric allocation problem with unit demand agents. 

\subsection{Single-Minded Bidding}\label{sec:single-minded}
As stated, a natural approach in trying to quantify the social welfare of a Nash equilibrium $v$ is to consider alternate strategies for the agents.
Of particular importance is {\em single-minded reporting} where an agent $i$ reports a value $1$ for a specific item $j$ and value $0$ for
every other item. We denote this report by $\SeqValue{j}$.

To analyze this change of strategy, let $\genTime{j}{v}$ be the minimum value of $t$ such that $\remAmm{j}{v}{t}=0$; recall this is the {\em consumption time} of item $j$ under the Nash equilibrium $v$. Now denote by
$\SingTime{j}=t_j\genValue{\SeqValue{j}}$, the consumption time of item $j$ when agent $i$ bids single-mindedly for $j$ and the other agents $-i$ report $v_{-i}$. 

Two properties of single-minded reporting will be useful. First, regardless of the strategies of the other agents $-i$, the consumption time
of item $j$ will be minimized if agent $i$ bids single-mindedly on item $j$. Second, at a Nash equilibrium, if agent $i$ deviates and bids single-mindedly on item $j$ whose consumption time was at most 1, then the  consumption time
of item $j$ can decrease by at most 75\%.

Before proving these two properties we remark that whilst the first property may seem self-evident 
there is a major subtlety due to dynamic knock-on effects. 
Indeed, when an item $\ell\neq j$ for which many agents have high value has been entirely consumed, 
the consumption rates get redistributed among the remaining items. 
It is necessary to show that bidding for $\ell$ does not decrease the completion time of $j$ despite leading to the other agents consuming more of $j$. 

The key to the proof is showing that these indirect
knock-on effects do not outweigh the direct effects of bidding single-mindedly.

\begin{lemma}\label{lem:minimal-t}
Given any $v_{-i}$, the consumption time of item $j$ is minimized when agent $i$ bids single-mindedly for $j$.
That is, $\min_{\overline{u}\in \valInd} t_j\genValue{\overline{u}}=\SingTime{j}$ and $\mathrm{arg}\hspace{-0.5mm}\min_{\overline{u}\in \valInd} t_j\genValue{\overline{u}}=\SeqValue{j}$
\end{lemma}

\begin{proof}
Take any agent $i\in I=[n]$, any item $j\in J=[m]$ and any $v\in \valGroup$. 
We remark that throughout the proof $v_{-i}$ will be fixed but starting from an arbitrary $v_i$ we will shift towards $\SeqValue{j}$.
Next, let $X=\set{j'\in [m]: v_i(j')>0,\ j'\neq j}$ be the set of items, other than $j$, for which 
agent $i$ reports a positive value.
Now label the items of $X$ by increasing consumption time; that is, $x_k$ for $k=1,\dots, |X|$ such 
that $\NETime{x_{k}}$ is increasing. 

The idea behind the proof is to construct a series of valuations $\{v_i=\IndValue{|X|},\IndValue{|X|-1},\dots, \IndValue{1}, 
\IndValue{0}=\SeqValue{j}\}$ such that $t_{j}\genValue{\IndValue{k-1}}\leq t_{j}\genValue{\IndValue{k}}$, for each $k\le |X|$, and where 
$\IndValue{k-1}$ has support of cardinality one less than $\IndValue{k}$.
Thus the consumption time of item $j$ is less with the single-minded report $\SeqValue{j}$ than with the report $v_i$.
Because the choice of $v_i$ was arbitrary the result will follow.

So we begin with  $\IndValue{|X|}=v_i$. Then, given $\IndValue{k}$, we define $\IndValue{k-1}$ as follows.
Let $\IndValue[x_{k}]{k-1}=0$, $\IndValue[j']{k-1}=\IndValue[j]{k}+\IndValue[x_{\ell}]{k}$ and 
$\IndValue[j']{k-1}=\IndValue[j']{k}$ for any other item $j'$.
For simplicity, we will use the notation $\IndValue{k}=\genValue{\IndValue{k}}$ 
and $t_j^{k}=t_{j}\genValue{\IndValue{k}}$.
Now consider $\IndValue{k-1}$ and $\IndValue{k}$, and let $T=\min \set{\IndTime{k}{x_k},\IndTime{k-1}{j}}$. 
If every item has the same consumption time in both $\IndValue{k-1}$ and under $\IndValue{k}$
then $\IndTime{k-1}{j}=\IndTime{k}{j}$, which implies that $\IndTime{k-1}{j}\leq \IndTime{k}{j}$ which is what we wanted.

Otherwise, let $t$ be 
the smallest time such that the set of items that have been consumed is different under both strategies. We wish to show that $t=T$. 
By definition, for any $\tau\in [0,t)$, we have $\remSet{(\IndValue{k})}{\tau}=\remSet{(\IndValue{k-1})}{\tau}$ because 
$j'$ is the first item to be consumed under one strategy but not the other. 
This implies that if either $i'\neq i$ or $j'\notin \set{j,x_{\ell}}$ then
we have $\conRate{i'}{j'}{\IndValue{k}}{\tau}=\frac{v_{i'}(j')}{\sum_{\ell'\in \remSet{(\IndValue{k})}{\tau}} v_{i'}(\ell')}=\conRate{i'}{j'}{\IndValue{k-1}}{\tau}$. 
Thus, if $j'\notin \set{j,x_{k}}$, the quantity of good $j$ remaining at time $t$ is the same for both strategies; that is, 
$\remAmm{j'}{\IndValue{k}}{t}=\remAmm{j'}{\IndValue{k-1}}{t}$, since the consumption integrals are identical in both cases. 
Consequently, it must be the case that $j'\in \set{j,x_{k}}$.

But the consumption rate $\conRate{i}{x_{k}}{\IndValue{k-1}}{\tau}=0<\conRate{i}{x_{k}}{\IndValue{k}}{\tau}$.
Hence, $\remAmm{x_{k}}{\IndValue{k}}{t} < \remAmm{x_{k}}{\IndValue{k-1}}{t}$. So, if $j'=x_{k}$ then its
consumption time must be smaller in $\IndValue{k}$ than in $\IndValue{k-1}$.
Similarly, the consumption rate $\conRate{i}{j}{\IndValue{k-1}}{\tau}=\conRate{i}{j}{\IndValue{k}}{\tau}+\conRate{i}{x_{k}}{\IndValue{k}}{\tau}>\conRate{i}{j}{\IndValue{k}}{\tau}$. Hence, $q_j^{(\IndValue{k'})}(t) > \remAmm{j}{\IndValue{k'-1}}{t}$. 
So, if $j'=j$ then its consumption time must be smaller in $\IndValue{k-1}$ than in $\IndValue{k}$.
This implies $t=T=\min \set{\IndTime{k}{x_k},\IndTime{k-1}{j}}$ as desired.
Thus, we have two cases to consider.

\textbf{\noindent{\tt Case I:}} $\IndTime{k-1}{j}\leq \IndTime{k}{x_{k}}$. Then $j$ must be the first item to be completed in $\IndValue{k-1}$ before being completed in $\IndValue{k}$. 
Hence, $\IndTime{k-1}{j}<\IndTime{k}{j}$ and the result holds.

\textbf{{\tt Case II:}} $\IndTime{k-1}{j}> \IndTime{k}{x_{k}}$.
Observe that $\remSet{(\IndValue{k})}{t}\subseteq \remSet{(\IndValue{k-1})}{t}$ for any time $t\in [\IndTime{k}{x_k}, \IndTime{k-1}{j}]$. 
Thus before $\IndTime{k-1}{j}$, no item can finish earlier under $\IndValue{k-1}$ than under $\IndValue{k}$.
This implies that, unless $j'=j$ and $i'=i$, we 
have $\conRate{i'}{j'}{\IndValue{k}}{t}\geq \conRate{i'}{j'}{\IndValue{k}-1}{t}$ until $j'$ is entirely consumed under $\IndValue{k}$. 
In particular, let $Y$ be the set of items with consumption time in the interval $[\IndTime{k}{x_{k}}, \IndTime{k}{j}]$ under $\IndValue{k}$.
(Note that $\{x_k, j\}\subseteq Y$.)
Then, only these items in $Y$ have a consumption time in $[\IndTime{k}{x_{k}},\IndTime{k-1}{j}]$ under $\IndValue{k-1}$. 
Now set $\sumInd{k}=\sum_{j'\in Y} \conRate{i}{j'}{\IndValue{k}}{t}$ and $\sumInd{k-1}=\sum_{j'\in Y} \conRate{i}{j'}{\IndValue{k-1}}{t}$. 
These are the total consumption rates at time $t$ of agent $i$ for items in $Y$ under under $\IndValue{k}$ and $\IndValue{k-1}$,
respectively. We claim that $\sumInd{k}\leq \sumInd{k-1}$ for $t\in [0,\IndTime{k-1}{j}]$.
To prove this, recall that, for $t\in [0,\IndTime{k}{j}]$, 
$\remSet{(\IndValue{k})}{t}\subseteq \remSet{(\IndValue{k-1})}{t}$.
Furthermore, $\remSet{(\IndValue{k-1})}{t}\setminus \remSet{(\IndValue{k})}{t}\subseteq Y$.
Denoting $\remSet{(\IndValue{k})}{t}=R_1$ and $\remSet{(\IndValue{k-1})}{t}=R_2$, and $\big(\sum\limits_{\ell\in R_2} v_{i}(j')\big)\big(\sum\limits_{j'\in R_1} v_{i}(j')\big)=R_3$ we obtain:

\begin{align*}
\lefteqn{\sumInd{k-1}-\sumInd{k}}\\
&\ =\ \sum\limits_{j'\in Y\cap R_2}\frac{v_{i}(j')}{\sum_{j^*\in R_2} v_{i}(j^*)}-\sum\limits_{j'\in Y\cap R_1}\frac{v_{i}(j')}{\sum_{j^*\in R_1} v_{i}(j^*)}\\
&\ =\ \frac{\sum\limits_{j'\in Y\cap R_2}v_{i}(j')}{\sum\limits_{j'\in R_2} v_{i}(j')}-\frac{\sum\limits_{j'\in Y\cap R_1}v_{i}(j')}{\sum\limits_{j'\in R_1}v_{i}(j')}\\
&\ =\ \frac{\left(\sum\limits_{j'\in Y\cap R_2}v_{i}(j')\right)\left(\sum\limits_{j'\in R_1} v_{i}(j')\right)-\left(\sum\limits_{j'\in Y\cap R_1}v_{i}(j')\right)\left(\sum\limits_{j'\in R_2} v_{i}(j')\right)}{\left(\sum\limits_{j'\in R_2} v_{i}(j')\right)\left(\sum\limits_{j'\in R_1} v_{i}(j')\right)}\\
&\ =\ \frac{\left(\sum\limits_{j'\in Y\cap (R_2\setminus R_1) }v_{i}(j') +\sum\limits_{j'\in Y\cap (R_1\cap R_2) }v_{i}(j')\right)
\sum\limits_{j'\in R_1} v_{i}(j')
-\sum\limits_{j'\in Y\cap R_1}v_{i}(j') \left(\sum\limits_{j'\in R_2\setminus R_1} v_{i}(j')+ \sum\limits_{j'\in R_1\cap R_2}v_{i}(j')\right)}{\left(\sum\limits_{j'\in R_2} v_{i}(j')\right)\left(\sum\limits_{j'\in R_1} v_{i}(j')\right)}\\
\end{align*}
But $R_1\subseteq R_2$ and $R_2\setminus R_1\subseteq Y$. So we have
\begin{align*}
\lefteqn{\sumInd{k-1}-\sumInd{k}}\\
&\ =\ \frac{\left(\sum\limits_{j'\in R_2\setminus R_1 }v_{i}(j') +\sum\limits_{j'\in Y\cap R_1 }v_{i}(j')\right)
\sum\limits_{j'\in R_1} v_{i}(j')
-\sum\limits_{j'\in Y\cap R_1}v_{i}(j') \left(\sum\limits_{j'\in R_2\setminus R_1} v_{i}(j')+ \sum\limits_{j'\in R_1}v_{i}(j')\right)}{\left(\sum\limits_{j'\in R_2} v_{i}(j')\right)\left(\sum\limits_{j'\in R_1} v_{i}(j')\right)}\\
&\ =\ \frac{\left(\sum\limits_{j'\in  R_2\setminus R_1 }v_{i}(j')\right)\left(\sum\limits_{j'\in R_1} v_{i}(j')\right)-\left(\sum\limits_{j'\in Y\cap R_1}v_{i}(j')\right)\left(\sum\limits_{j'\in R_2\setminus R_1} v_{i}(j')\right)}{\left(\sum\limits_{j'\in R_2} v_{i}(j')\right)\left(\sum\limits_{j'\in R_1} v_{i}(j')\right)}\\
&\ =\ \frac{\left(\sum\limits_{j'\in  R_2\setminus R_1 }v_{i}(j')\right)\left(\sum\limits_{j'\in R_1\setminus Y} v_{i}(j')\right)}{\left(\sum\limits_{j'\in R_2} v_{i}(j')\right)\left(\sum\limits_{j'\in R_1} v_{i}(j')\right)}\\
&\ \geq\ 0
\end{align*}
By definition of $Y$, every item in $Y$ has consumption time earlier than $j$ under $\IndValue{k}$.
However, before the consumption time $\IndTime{k}{j}$ of $j$ under $\IndValue{k}$, 
the total consumption rate of all the items in $Y$ is at least as large under $\IndValue{k-1}$ than 
under $\IndValue{k}$. 
Therefore at consumption time $\IndTime{k-1}{j}$ of $j$ under $\IndValue{k-1}$ not every item of $Y$ has been
consumed under $\IndValue{k}$. In particular, $j$ itself has not 
been consumed by then under $\IndValue{k}$. So, $j$ is consumed faster in $\IndValue{k-1}$ 
than in $\IndValue{k}$.

We iterate this argument 
with $k-1$ until we get $k=0$ and we have computed $\IndValue{0}=\SeqValue{j}$ in which only $j$ has non-0 value.
This 
implies that $\IndValue{0}=\SeqValue{j}$ as desired.
\end{proof}

The consumption time of each item in the Nash equilibrium will be denoted as the time $\NETime{j}=t_j(v)$. For convenience we denote $\NETime{0}=0$.
Without loss of generality, we relabel the items so that $\NETime{j}$ is increasing with $j$.

A second property which we require for single-minded bidding is the following, which is an extension of a lemma from~\cite{CFF15}.

\begin{lemma}\label{lem:Single-target}
Let $v$ be a pure Nash equilibrium. Take any agent $i$ and any item $j$ whose consumption time is at most $1$. 
The consumption time of item $j$ decreases by at most 75\% if agent $i$ switches to the single-minded
strategy $\SeqValue{j}$, that is if $\NETime{j}\leq 1$ then $\SingTime{j}\geq \frac{1}{4}\cdot \NETime{j}$.
\end{lemma}
\begin{proof}
Assume that $\sum_{i'\neq i} \conRate{i'}{j}{\SeqValue{j}}{t}<1$ for any time $t$.
Then $\fn{\remAmm{i}{\SeqValue{j}}{t}}{\frac{1}{2}}>0$ 
which implies that $\SingTime{j}>\frac{1}{2}\geq \frac{1}{4} \NETime{j}$. 
So we may assume there is a smallest time $\tau$ such that $\sum_{i'\neq i} \conRate{i'}{j}{\SeqValue{j}}{\tau}\geq 1$. 
We now have two cases.\\

\noindent{\tt Case I:} $\tau\geq \frac{1}{4} \NETime{j}$\\
By definition, the total consumption rate of item $j$ is positive at time $\tau$.
Thus, item $j$ is still available at time $\frac{1}{4} \NETime{j}$. Consequently
$\SingTime{j}\geq \frac{1}{4} \NETime{j}$.\\

\noindent{\tt Case II:} $\tau < \frac{1}{4} \NETime{j}$\\
Agent $i$ has a consumption rate $1$ for item $j$ under $\SeqValue{j}$ until its consumption time.
Note that before its consumption time the total consumption rate for $j$ is non-decreasing.
In particular, before $\tau$ (phase 1) the total consumption rate of the other agents for $j$ is at most $1$.
After $\tau$ (phase 2) the total consumption rate of the other agents for $j$ is at least $1$.

Then agent $i$ consumes $\tau$ units of good $j$ in phase 1. In phase 2,  the other agents consume $j$ at least as fast as $i$. 
Thus agent $i$ consumes at most half the remaining amount of good $j$, which is obviously at most $\frac{1}{2}$. 
So agent $i$ gets at most $\frac{1}{4} \NETime{j}+\frac{1}{2}$ units of good $j$.
This implies the other agents get at least $\frac{1}{2}-\frac{1}{4}\NETime{j}$ units of good $j$. 

Recall from the proof of Lemma~\ref{lem:minimal-t} that
$\remSet{(v)}{t} = \remSet{(\IndValue{|X|})}{t}\subseteq\remSet{(\IndValue{0})}{t}=\remSet{(\hat{u}_j)}{t}$ for any time $t\le t_j(\hat{u}_j, v_{-i})$,
that is before the consumption time of item $j$ under $\SeqValue{j}$.
This implies that, at each point in time, the total consumption rate at which the agents (excluding $i$) consume item $j$ 
is smaller under $\SeqValue{j}$ than under $v$. 
In particular, if it takes time $t(\alpha)$ for the agents excluding $i$ to consume $\alpha$ units of item $j$ under $v$ then it will take at 
least $t(\alpha)$ for them to consume $\alpha$ units of $j$ under $\SeqValue{j}$. Moreover, recall $\conRate{i'}{j}{\NEValue}{t}$ is non-decreasing. 
Thus, for any $\beta\in (0,1)$, the time it takes the agents excluding $i$ to consume $\beta\cdot \alpha$ units of good $j$ under $v$ is at least 
$\beta\cdot t(\alpha)$. 

Furthermore, if $\alpha$ is the amount of good $j$ that the agents excluding $i$ consume under $v$ then $t(\alpha)=\NETime{j}(v)$. 
Now set $\beta=\frac{\frac{1}{2}-\frac{1}{4}\NETime{j}(v)}{\alpha}$. Then because $\frac{1}{\alpha}\ge 1$ and $\frac14 \NETime{j}(v) \le \frac14$ 
we have that $\beta\ge \frac{1}{4}$. Moreover, since $\alpha\ge \frac{1}{2}-\frac{1}{4}\NETime{j}(v)$, we have $\beta\le 1$. 
Hence the agents excluding $i$ consume at least $\alpha\cdot \beta =\frac{1}{2}-\frac{1}{4}\NETime{j}(v)$ units in time $\NETime{j}(\SeqValue{j})$
under $\SeqValue{j}$ and they consume $\alpha$ units in time $\NETime{j}(v)$ under $v$.
Thus
$$\NETime{j}(\SeqValue{j})\geq t(\beta\alpha) \ge \beta t(\alpha) = \beta \NETime{j}(v)$$
and since $\beta\ge \frac14$ we have:
$$\beta \NETime{j}(v)\geq \frac{1}{4} \NETime{j}(v).$$
Thus $\NETime{j}(\SeqValue{j})\geq \frac{1}{4} \NETime{j}(v)$ as desired.
\end{proof}

\subsection{Sequential Bidding}\label{sec:sequential}

Unfortunately, consideration of deviations to single-minded bidding strategies is insufficient to prove a good price of anarchy bound for the Cardinal Probabilistic Serial mechanism.
Indeed, this is intuitively obvious. If an agent wins many items in the optimal solution to the allocation problem then a strategy that targets a single item will likely to do very poorly in comparison.

To circumvent this problem, we consider a second class of strategies, which we term {\em sequential bidding}.
The idea is that an agent has a target set $X=\{x_1, x_2, \dots, x_k\}$ and,
moreover, requests to consume the items one-at-a-time in the given order.
However, the Cardinal Probabilistic Serial is not perfectly compatible with such a sequential request. But it does allow the agents to mimic such a strategy with arbitrary precision. To see this, given a finite sequence $X=\set{x_1,\dots,x_k}$ and
$\varepsilon\in \opInt{0}{\frac{1}{2}}$, define the {\em epsilon-valuation} $\EpsValue{X}$ to be $1-\sum_{\ell=1}^{k-1} \varepsilon^{\ell}$ if $j=x_1$,  $\varepsilon^\ell$ if $j=x_{\ell}$ for $\ell>1$, and 0 if $j\notin X$.

Then, given a finite sequence $X$, we can define the {\em sequential bidding strategy} $\SeqValue{X}=\lim_{\varepsilon\rightarrow 0} \EpsValue{X}$.
Under this sequential bidding strategy, at any time $t$, the agent will consume the first item in $X$ that has not yet been entirely consumed. 
({\em We defer a formal mathematical justification for the validity of this
construction to Appendix~\ref{sec:epsilon}.})

Using the two properties we obtained for single-minded bidding, 
we can analyse the consequences of deviating to a sequential bidding strategy. Specifically, we prove the following technical lemma.

\begin{lemma}\label{lem:lower-bound}
For any agent $i$, let $v'_i$ be the true value $i$ has for the items and let $v$ be any pure Nash equilibrium with respect to $v'$. 
Then, for any sequence of items $X=\set{x_1,x_2,\dots,x_k}$ whose consumption times are bounded above by $1$, it holds 
that: \[v'_i(CPS(v))\ \geq\ \frac{1}{4}\cdot \sum_{\ell=1}^k (\NETime{x_{\ell}}-\NETime{x_{\ell-1}})\cdot \TrueValue[i](x_{\ell}).\]
\end{lemma}
\begin{proof}
Recall that we have:
\begin{itemize}
    \item $\NETime{\ell}$: consumption time of item $x_{\ell}$ in the Nash equilibrium $v$.
    \item $\SingTime{\ell}$: consumption time of item $x_{\ell}$ in $\genValue{\SeqValue{X}}$ where $i$ makes a sequential bid for $X$.
\end{itemize}

We additionally denote the consumption time of item $j$ when an agent switches to a sequential strategy as $\SeqTime{j}=t_j \genValue{\SeqValue{X}}$. Moreover, for convenience we denote $\SeqTime{0}=0$.
    
By Lemma~\ref{lem:minimal-t} we have $\SingTime{j}\geq \SeqTime{\ell}$. By Lemma~\ref{lem:Single-target} we 
have $\SeqTime{\ell}\geq \frac{1}{4} \NETime{\ell}$. Therefore, $\SingTime{j}\geq \frac{1}{4} \NETime{\ell}$. 
Now recall, under the strategy $\SeqValue{X}$, the items of $X=\set{x_1,x_2,\dots,x_k}$ are ordered in decreasing 
order of value for agent $i$. This means that before time $\SingTime{\ell}$ the agent consumes an item whose value is at least
$\TrueValue[i](x_{\ell})$. In particular, because $\SingTime{j}\geq \frac{1}{4} \NETime{j}$, 
if $\NETime{\ell-1}\le \NETime{\ell}$
then during the interval $[\frac{1}{4} \NETime{\ell-1},\frac{1}{4} \NETime{\ell}]$ agent $i$ consumes an item of value at least
$\TrueValue[i](x_{\ell})$. Therefore, agent $i$ has a payoff of
\begin{eqnarray*}
\lefteqn{CPS(\TrueValue; \{\SeqValue{X}, v_{-i}\}) 
\geq \sum_{\ell=1}^k \, \max \left(\frac{1}{4} \NETime{\ell}-\frac{1}{4} \NETime{\ell-1}\, , \, 0\right)\cdot \TrueValue[i](x_\ell)}\\
&\geq \sum_{\ell=1}^k  \left(\frac{1}{4} \NETime{\ell}-\frac{1}{4} \NETime{\ell-1}\right) \cdot \TrueValue[i](x_\ell)
\geq \frac{1}{4}\cdot \sum_{\ell=1}^k \left(\NETime{\ell}-\NETime{\ell-1}\right)\cdot \TrueValue[i](x_\ell)
\end{eqnarray*}
As $v$ is a Nash equilibrium, it must also give agent $i$ a payoff of at least 
$\frac{1}{4}\cdot \sum_{i=1}^k \left(\NETime{\ell}-\NETime{\ell-1}\right)\cdot \TrueValue[i](x_\ell)$. 
\end{proof}

\section{The Price of Anarchy}\label{sec:POA}

We are now ready to quantify the price of anarchy.
We begin with the upper bound for Cardinal Probabilistic Serial in Section~\ref{sec:UB}, followed by the upper bound for Probabilistic Serial in Section~\ref{sec:UB-PS}. We will then present complementary lower bounds in Section~\ref{sec:LB}.

\subsection{An Upper Bound on the Price of Anarchy for Cardinal Probabilistic Serial}\label{sec:UB} 
To give an upper bound on the price of anarchy, we proceed in two steps. In the first step, we will assume that the number of agents and items is the same, that is we consider the {\em symmetric case} where $m=n$ and prove that the price of anarchy is at most $O(\sqrt{n}\cdot \log n)$ with multi-unit demand agents. Then we extend this result to the {\em asymmetric case}, where $m$ is arbitrary, and show the price of anarchy is $O(\sqrt{n}\cdot \log m)$.

\subsubsection{The Symmetric Case.}\label{sec:symmetric}

Throughout this section, we will assume that $m=n$. Observe that this implies that the completion time of each 
item is at most 1, allowing us to apply Lemma~\ref{lem:lower-bound} for all items. 

We now formulate the price of anarchy
as an optimization program. This optimization program is very difficult to handle directly. So our basic approach will be to apply a series of relaxations and simplifications until we obtain a program we can solve.
The task is to ensure the transformations are consistent with generating upper bounds and that they do not degrade the value of the objective function excessively. 

We show the following key result which will be used to prove one case of the main result.

\begin{theorem}\label{thm:sym}
    In the symmetric one-sided allocation problem with multi-unit demand agents, let $OPT$ be the social welfare of the optimal allocation. The social welfare of any Nash equilibrium is at least $\fn{\Omega}{\frac{OPT}{\sqrt{n}\cdot \log n}}$ for Cardinal Probabilistic Serial.
\end{theorem}

\begin{proof}

The bidding strategies that will be used throughout the section, single-minded bidding and sequential bidding, are applicable to the ordinal version of Probabilistic Serial. So, the bounds on the value the agents get in the Nash Equilibrium by using these strategies for Cardinal Probabilistic Serial will extend to the usual Probabilistic Serial as the Probabilistic Serial strategy used by the other agents can be used for Cardinal Probabilistic Serial. For the rest of the proof, we will work with Cardinal Probabilistic Serial.

Let $\{X_1, X_2, \dots, X_n\}$ be the optimal allocation,
where each agent $i$ receives the bundle of items $X_i=\set{x_1^{i},\dots,x_{k_i}^{i}}$.
Here we assume the items in $X_i$ are ordered by increasing 
consumption time in the Nash equilibrium $v$.
We can then use Lemma~\ref{lem:lower-bound} to lower bound the social welfare of
the Nash equilibrium $v$. To do this first note that, whilst the items of
$X_i$ are ordered by consumption time they are not ordered by value.
In particular, for the lower bound in Lemma~\ref{lem:lower-bound} 
we may use the right-to-left maxima of $\{v'_i(x^i_{1}), v'_i(x^i_{2}), \dots, v'_i(x^i_{k_i})\}$.
This gives a lower bound on the social welfare of the Nash equilibrium $v$ of:
\begin{equation*}
\frac14\sum_{i=1}^{n}\ \sum_{\ell=1}^{k_i}\left( \max_{\ell'=\ell,\dots,k_i}\set{v'_{i}(x_{\ell'})}\cdot (\NETime{x_{\ell}^i}(v)-\NETime{x_{\ell-1}^i}(v))\right)
\end{equation*}
We may then bound the price of anarchy using the following optimization program:
\begin{eqnarray}
\min \hspace{.5cm} \frac{\frac14\cdot\sum_{i=1}^{n} \sum_{\ell=1}^{k_i}\left( \max_{\ell'=\ell,\dots,k_i}\set{v'_{i}(x_{\ell'})}\cdot (\NETime{x_\ell^i}(v)-\NETime{x_{\ell}^i}(v))\right)}{OPT}\label{opt:obj}\\
\text{s.t.}\hspace{2cm} \sum_{i=1}^n \sum_{\ell=1}^{k_i} v'_{i,x_\ell^i}=& OPT&\label{opt:OPT}\\
   \bigcup_{i=1}^n \set{x_\ell^i:\ell\in [k_i]}=&[n]\label{opt:partition}\\
    \sum_{j=1}^n \TrueValue[i](j) =& 1&\forall i\in [n]\label{opt:unit-v'}\\
 \sum_{j=1}^n \NEValue[i](j) =& 1&\forall i\in [n]\label{opt:unit-v}\\
    \NETime{j}(v)\leq& \NETime{j+1}(v)&\forall j\in[n-1]\label{opt:time}\\
       v\in &NE(v')&\label{opt:NE}
\end{eqnarray}

Let's understand this optimization program. 
Constraints (\ref{opt:unit-v'}) and (\ref{opt:unit-v}) state that for every agent $v'_i$ and $v_i$ are unit-sum valuation functions. 
Constraints (\ref{opt:OPT}) and (\ref{opt:partition}) ensure $\{X_1, X_2, \dots, X_n\}$ is a partition of the items with optimal social welfare (with respect to the true valuation functions $v'$).
Next the constraint~(\ref{opt:time}) forces the items to be ordered by increasing consumption time. Finally, the constraint~(\ref{opt:NE})
states that $v$ is a Nash equilibrium with respect to the true valuations $v'$.
The objective function~(\ref{opt:obj}) then gives a worst-case bound
on the price of anarchy using Lemma~\ref{lem:lower-bound}.

However, this optimization program is difficult to analyze so our task now is to simplify the program without weakening the resultant price of anarchy guarantee.
To do this, our first step is to fix $OPT$, the social welfare of the optimal solution. (We will later determine the worst case values of $OPT$.)
In doing so we may omit the denominator from the objective function~(\ref{opt:obj}).
Second, observe that the bound can only be worse if we relax or remove some of the constraints from the optimization program.
In particular, let's omit the Nash equilibrium constraint~(\ref{opt:NE}). This gives:
\begin{eqnarray}
\min \hspace{.5cm} \frac14\cdot\sum_{i=1}^{n} \sum_{\ell=1}^{k_i}\left( \max_{\ell'=\ell,\dots,k_i}\set{v'_{i}(x_{\ell'})}\cdot (\NETime{x_\ell^i}(v)-\NETime{x_{\ell}^i}(v))\right) \label{opt:OPT'}\\
\text{s.t.}\hspace{2cm} \sum_{i=1}^n \sum_{\ell=1}^{k_i} v'_{i,x_\ell^i}=& OPT&\nonumber\\
   \bigcup_{i=1}^n \set{x_\ell^i:\ell\in [k_i]}=&[n]\nonumber\\
    \sum_{j=1}^n \TrueValue[i](j) =& 1&\forall i\in [n]\nonumber\\
 \sum_{j=1}^n \NEValue[i](j) =& 1&\forall i\in [n]\nonumber\\
    \NETime{j}(v)\leq& \NETime{j+1}(v)&\forall j\in[n-1]\nonumber
\end{eqnarray}
The reader may ask if removing the Nash equilibrium constraint~(\ref{opt:NE}) will then
render the optimization program useless.
As we will see, the answer is no because implicitly the Nash equilibrium
conditions have been used in deriving the objective function.
Now note that
$\max_{\ell'=\ell,\dots,k_i}\set{v'_{i}(x_{\ell'})} \ge v'_{i}(x_{\ell})$.
Thus after removing the Nash equilibrium constraint we may now
assume in the worst case that the items of $X_i$ are also 
ordered in decreasing value. That is, the items of $X_i=\set{x_1^{i},\dots,x_{k_i}^{i}}$ decrease in 
both consumption time and in value. Adding this new constraint~(\ref{opt:value}) then gives the program
\begin{eqnarray}
\min \hspace{.5cm} \frac14\cdot\sum_{i=1}^{n} \sum_{\ell=1}^{k_i}\left( v'_{i}(x_{\ell})\cdot (\NETime{x_\ell^i}(v)-\NETime{x_{\ell}^i}(v))\right) \label{opt:OPT''}\\
\text{s.t.}\hspace{2cm} \sum_{i=1}^n \sum_{\ell=1}^{k_i} v'_{i,x_\ell^i}=& OPT&\nonumber\\
   \bigcup_{i=1}^n \set{x_\ell^i:\ell\in [k_i]}=&[n]\nonumber\\
    \sum_{j=1}^n \TrueValue[i](j) =& 1&\forall i\in [n]\nonumber\\
 \sum_{\ell=1}^{k_i} \NEValue[i](x^i_{\ell}) \le& 1&\forall i\in [n]\label{opt:less-unit}\\
    \NETime{j}(v)\leq& \NETime{j+1}(v)&\forall j\in[n-1]\nonumber\\
     \NEValue[i](x^i_\ell)\geq& \NEValue[i](x^i_{\ell+1})&\forall i\in [n], \forall \ell \in[k_i]\label{opt:value}
\end{eqnarray}
Note above that we may replace the unit-sum condition~(\ref{opt:unit-v}) on $v_i$ 
by a constraint~(\ref{opt:less-unit}) only on the values of items in $X_i$.

For the next step, for each agent $i$ we use a change of variables to
$\{y_1^i,y_2^i\dots,y_{k_i}^i\}$.
Specifically, set $y^i_{k_i}=k_i\cdot v'_{i}(x^{i}_{k_i})$.
Then, recursively, set $y^i_{\ell}= \ell\cdot \left(v'_{i}(x^{i}_{\ell})-v'_{i}(x^{i}_{\ell+1})\right)$, for each $\ell=\{k_i-1,\dots, 2,1\}$. 
Observe that 
\begin{equation*}
\sum_{\ell=1}^{k_i} y_\ell^i
\ =\  \sum_{\ell=1}^{k_i} k_i\cdot v'_{i}(x^{i}_{k_i})
\ =\ \sum_{\ell=1}^{k_i} v'_{i}(x_\ell^i)   
\end{equation*}
In particular, the $y^i_\ell$ are non-negative and sum to at most one.
Moreover, we have $v'_{i}(x^{(i)}_\ell)=\sum_{\ell'=\ell}^k \frac{y^i_{\ell'}}{\ell'}$. 
Thus we have
\begin{align*}
    \sum_{i=1}^{n} \sum_{\ell=1}^{k_i}\left( v'_{i,x_\ell^i}\cdot \left(\NETime{x_{\ell}^{i}}(v)-\NETime{x_{\ell-1}^i}(v)\right)\right)
    &=\sum_{i=1}^{n} \sum_{\ell=1}^{k_i}\left( \sum_{\ell'=\ell}^{k_i} \frac{y_{\ell'}^i}{\ell'}\cdot \left(\NETime{x_{\ell}^{i}}-\NETime{x_{\ell-1}^i}\right)\right)\\
    &=\sum_{i=1}^{n}\sum_{\ell'=1}^{k_i}\frac{y_{\ell'}^i}{\ell'} \sum_{\ell=1}^{\ell'} \left(\NETime{x_{\ell}^{i}}-\NETime{x_{\ell-1}^i}\right)\\
    &=\sum_{i=1}^{n} \sum_{\ell'=1}^{k_i}\left( \frac{y_{\ell'}^i}{\ell'}\cdot \NETime{x_{\ell'}^i}\right)
    \end{align*}
So, relabelling $\ell'$ as $\ell$, we obtain the optimization program:
\begin{eqnarray}
\min \hspace{.5cm} \sum_{i=1}^{n} \sum_{\ell=1}^{k_i}\left( \frac{y_\ell^i}{\ell}\cdot \NETime{x_{\ell}^i}\right)&&\label{opt:obj2}\\
\text{s.t.}\hspace{1.75cm} \sum_{i=1}^n \sum_{\ell=1}^{k_i} y_\ell^i=& OPT&\label{opt:y}\\
   \bigcup_{i=1}^n \set{x_\ell^i:\ell\in [k_i]}=&[n]\nonumber\\
    \sum_{\ell=1}^{k_i} y_\ell^i \leq& 1&\forall i\in [n]\\
    y_\ell^i \geq& 0&\forall i\in [n], \forall \ell\in [k_i]\\
     \NETime{j}(v)\leq& \NETime{j+1}(v)&\forall j\in[n-1]\nonumber
\end{eqnarray}
Observe above that, for simplicity, we have removed the factor $\frac{1}{4}$ from the objective function~(\ref{opt:obj2}). We will reincorporate it later. 

Let's now investigate the structure of the optimal solution to this program.
We claim that, for each agent $i$, only one $y^i_\ell$ need be positive.
To see this, assume $y_\ell^i>0$ and $y_{\ell'}^i>0$ for $\ell\neq \ell'$. 
Without loss of generality, let $\frac{\NETime{x_{\ell}^{i}}}{\ell}\leq \frac{\NETime{x_{\ell'}^{i}}}{\ell'}$. Then 
replacing $y_\ell^i$ by $y_\ell^i+y_{\ell'}^i$ and replacing $y_{\ell'}^i$ 
by $0$ decreases (or keeps constant) the objective function. 
We may enforce this by adding a constraint denoting $\ell_i$ to be the 
index which minimizes $\frac{\NETime{x_{\ell}^i}}{\ell}$. For convenience, relabel $\NETime{x_{\ell_i}^i}$ as $t_i$ and $y_{\ell_i}^i$ as $y_i$.
\begin{eqnarray}
\min \hspace{.5cm} \sum_{i=1}^{n} \left( y_i\cdot \frac{t_i}{\ell_i}\right)&&\label{opt:OPT'''}\\
\text{s.t.}\hspace{1.75cm} \sum_{i=1}^n y_i=& OPT&\nonumber\\
 \bigcup_{i=1}^n \set{x_\ell^i:\ell\in [k_i]}=&[n]\nonumber\\
 \ell_i =& \mathrm{arg}\hspace{-0.5mm}\max_{\ell \in [k_i]} \frac{t_i}{\ell} &\forall i \label{opt:li}\\
    y_i \leq& 1&\forall i\in [n]\label{opt:li1}\\
    y_i \geq& 0&\forall i\in [n]\label{opt:li0}\\
     \NETime{j}(v)\leq& \NETime{j+1}(v)&\forall j\in[n-1]\nonumber
\end{eqnarray}
We may now apply a similar trick over pairs of agents.
Assume there are two agents $i,i'$ and two indices $\ell_i, \ell_{i'}$
such that both $y_i$ and $y_{i'}$ are non-integral,
that is, $y_i\in (0,1)$ and $y_{i'}\in (0,1)$. 
Without loss of generality, let $\frac{t_i}{\ell_i}\leq \frac{t_{i'}}{\ell_{i'}}$.
Then replacing $y_i$ by $y_i+\delta$ and replacing $y_{i'}$ 
by $y_{i'}-\delta$, for some small $\delta$, decreases (or keeps constant) the objective function. 
Note that this is a feasible change because constraint~(\ref{opt:y}) is remains tight and constraints~(\ref{opt:li1}) and~(\ref{opt:li0}) still hold.
Moreover, setting $\delta=\min\set{1-y_i, y_{i'}}$ forces either $y_i$ or 
$y_{i'}$ to become integral.
But this implies there is an optimal solution in which exactly one $y_i$
is non-integral.

In particular, let $k=\floor{OPT}$. Then we may relabel the agents so that
$y_i= 1$ for each $1\le i \le k$, 
$y_{k+1}= OPT-k$, and $y_i= 0$ for each $k+2\le i \le n$. 
Thus our problem simplifies to:
\begin{eqnarray*}
\min \hspace{.5cm} \sum_{i=1}^{k} \frac{t_i}{\ell_i}+
(OPT-k)\cdot \frac{t_{k+1}}{\ell_{k+1}}
&&\\
\text{s.t.}\hspace{1.5cm} \bigcup_{i=1}^n \set{x_\ell^i:\ell\in [k_i]}=&[n]
\end{eqnarray*}
Of course, we can further reduce the objective function by removing its second term. This gives:
\begin{equation}
    \begin{aligned}
    \min \hspace{.5cm} \sum_{i=1}^{k} \frac{t_i}{\ell_i}\\
    \text{s.t.}\hspace{.5cm} \bigcup_{i=1}^n \set{x_\ell^i:\ell\in [k_i]}&=[n]
    \end{aligned}
    \label{eq:final-opt}
\end{equation}

To evaluate (\ref{eq:final-opt}), recall that the items are labelled in increasing order of consumption time. These consumption times then satisfy the following property.
\begin{lemma}\label{lem:time}
The consumption time of item $j$ must satisfy $t_j\ge \frac{j}{n}$.
\end{lemma}
\begin{proof}
Each agent has a total consumption rate of $1$.
Consequently, the total consumption rate of all agents is $n$.
Thus at time $t=\frac{j}{n}$ the number of units consumed of all goods is exactly $j$. But the quantity of each good is each is exactly $1$, so at most $j$ goods can have been completely consumed at time $t=\frac{j}{n}$. Hence the consumption time of good $j$ is $t_j\ge \frac{j}{n}$.
\end{proof}

Next we partition the agents into groups depending upon their $\ell_i$.
Specifically, let $I_\tau=\set{i\in [k]: \ell_i\in [2^\tau,2^{\tau+1})}$ for all $0\le \tau \le \ceil{\log n}$. 
Further, for each agent $i\le k$ we let $T_i$ be the consumption time of $x^i_{\ell_i}$.
We then order the agents of $I_\tau$ by increasing $T_i$. We use the notation
$i^\tau_q$ to denote the $q^{th}$ agent of $I_\tau$ in this ordering. 

In particular, by the time $x^i_{\ell_i}$ is consumed for $i=i^\tau_q$, at least $2^{\tau}\cdot q$ items have been consumed. Thus, by Lemma~\ref{lem:time}, the consumption time
of this item is 
$t_{{i_{q}^\tau}}\ge \frac{2^\tau\cdot q}{n}$. In particular,
\begin{equation}
\frac{t_{{i_{q}^\tau}}}{\ell_{i_q^\tau}}
\ \geq \ \frac{2^\tau\cdot q}{n}\cdot\frac{1}{\ell_{i_q^\tau}}
\ \geq \ \frac{2^\tau\cdot q}{n}\cdot\frac{1}{2^{\tau+1}}
\ =\ \frac{q}{2n}  
\end{equation}
We can now obtain a useful bound on the value of the optimization program.
\begin{eqnarray*}
\sum_{i=1}^{k} \frac{t_i}{\ell_i}
&\geq& \sum_{\tau=0}^{\ceil{\log n}} \sum_{q=1}^{|I_\tau|} \frac{q}{2n}\\
&=&\sum_{\tau=0}^{\ceil{\log n}} \frac{|I_\tau| \cdot (|I_\tau|+1)}{4n}\\
&\ge& \frac{1}{4n}\cdot \sum_{\tau=0}^{\ceil{\log n}} |I_\tau|^2\\
&\ge& \frac{1}{4n}\cdot \max_{0\le \tau \le \ceil{\log n}} |I_\tau|^2\\
&\geq& \frac{1}{4n}\cdot \left(\frac{k}{\log n+1}\right)^2\\
&=& \frac{k^2}{4n\log^2 n+\fn{o}{n\log n}}
\end{eqnarray*}
We are now ready to prove our price of anarchy upper bound.
By applying Lemma~\ref{lem:lower-bound}, with $X=[n]$, at any Nash equilibrium each agent is guaranteed
a payoff of at least $\frac{1}{4n}$. Thus the social welfare of any Nash equilibrium is at least $\frac14$.
The price of anarchy is then at most
\begin{align*}
  \max_{k\in[1,n]} &\big\{\min\{\frac{k}{\frac14}\ ,\ \frac{k}{\frac{1}{4}\cdot \frac{k^2}{4n \log^2 n+\fn{o}{n\log n}}}\}\big\}\\
    &\leq \max_{k\in[1,n]} \set{\min\set{4k\ ,\ (16n \log^2 n+\fn{o}{n\log n})/k}}\\
    &= 2\sqrt{n}\cdot \log n+\fn{o}{\sqrt{n\log n}}
\end{align*}
So, the price of anarchy of the Cardinal Probabilistic Serial mechanism is $\bigO{\sqrt{n}\cdot \log n}$ in the symmetric case, even with multi-unit demand agents. .
\end{proof}

\subsubsection{The Asymmetric Case.}\label{sec:asymmetric}

To extend our proof of the upper bound to the asymmetric setting, we proceed by case analysis. The first case is when 
the items which are consumed before time $1$ make the largest contribution to social welfare. 
In this case, the proof used to show Theorem~\ref{thm:sym} can be used to show an upper bound of $O(\sqrt{n}\cdot \log m)$. 
In the second case, when the items which are consumed after time $1$ make the largest contribution to social welfare, 
we show there is then a small set of items that make a significant contribution to social welfare; 
moreover, the agents have a strategy to win these items with constant probability. 

To prove our main positive result, we begin with the following result concerning the first case.

\begin{theorem}\label{thm:main-case-1}
    Let $X$ be the set of items whose consumption time is at most $1$. If $OPT_X$ is the social welfare of the optimal allocation of these items, then the social welfare of any Nash equilibrium is at least $\fn{\Omega}{\frac{OPT_X}{\sqrt{n}\cdot \log |X|}}$ for Cardinal Probabilistic Serial.
\end{theorem}

\begin{proof}
    The proof of Theorem~\ref{thm:sym} for the symmetric case applies in bounding the value the agents get in a Nash equilibrium using only items in $X$. However, we are not guaranteed to have $n$ items in $X$, instead we have $|X|\leq n$.\footnote{We remark that Lemma~\ref{lem:time} implies that $|X|\leq n$.} So, denoting $k=|X|$, when we simplify the optimization program 
    we get the following
    \begin{align*}
        \min \hspace{.5cm} \sum_{i=1}^{k} \frac{t_i}{\ell_i}\\
        \text{s.t.}\hspace{.5cm} \bigcup_{i=1}^n \set{x_\ell^i:\ell\in [k_i]}&=[k]
    \end{align*}
    
    We can then mimic the rest of the proof of Theorem~\ref{thm:sym} and split the items into $\log k$ sets containing items $2^\tau$ to $2^{\tau+1}-1$. Replacing $\log n$ by $\log k$ in the proof, we get a bound of $2\sqrt{n}\cdot \log k+\fn{o}{\sqrt{n\log k}}$, as required.
    \qed
\end{proof}

In order to prove the second case, neither the single-minded bidding strategy nor the sequential bidding strategy is sufficient. We need an additional alternate strategy which we use on items with high consumption times. To deal with these items, we will take advantage of the flexibility we get from having a cardinal mechanism and use a {\em uniform bidding} strategy. With this uniform bidding strategy, agents will report a value of $\frac{1}{k}$ for each item in a set $\set{x_1,x_2,\dots,x_k}$.

\begin{theorem}\label{thm:main-case-2}
    For Cardinal Probabilistic Serial, for any agent $i$, let $v'_i$ be the true value $i$ has for the items and let $v$ be any pure Nash equilibrium with respect to $v'$. 
    Then, for any sequence of items $X=\set{x_1,x_2,\dots,x_k}$ whose consumption times are bounded below by $q\in \bb{N}$, for some $k\leq q$ it holds that: 
    \[v'_i(CPS(v))\ \geq\ \frac{1}{2}\cdot \sum_{\ell=1}^k \TrueValue[i](x_{\ell}).\]
\end{theorem}

\begin{proof}
    Since all the items have a consumption time at least $q$ and the consumption rate is non-decreasing, the amount which has been consumed at time $\frac{q}{2}$ must be at most half of the amount left at time $q$, which is at most 1. This implies that at time $\frac{q}{2}$, at most half of each item in $X$ has been consumed.

    Now consider the uniform bidding strategy on $X$. Applying the same argument as for the single-minded strategy, no item will be consumed faster under the uniform strategy until an item from $X$ has been consumed. At time $\frac{k}{2}$, unless the remaining agents have consumed more than $\frac{1}{2}$ of some item of $X$, the agent $i$ who switches to the uniform bidding strategy has consumed exactly $\frac{1}{2}$ of each item in $X$. Up to time $\frac{q}{2}$, the remaining agents consume items of $X$ slower than in the Nash equilibrium. Namely, since the consumption time of each item of $X$ is at least $q$, the agents consume at most half of each item by time $\frac{q}{2}$. In particular, since $k\leq q$ they consume at most half of each item by time $\frac{k}{2}$.

    Now agent $i$ consumes exactly $\frac{1}{2}$ a unit of each of these items while the remaining agents consume no more than $\frac{1}{2}$ a unit of each item in total at time $\frac{k}{2}$. Thus $i$ consumes exactly half of the items of $X$ by time $\frac{k}{2}$ using the uniform bidding strategy. Consequently, agent $i$ obtains a value at least half of its total value for $X$, as desired. 
\end{proof}

We can now prove our main positive result.

\CPSAnarchy*

\begin{proof}
First, assume that $OPT\leq 4\sqrt{n}\cdot \log m$. Then, by Lemma~\ref{lem:time}, up to time $\frac{j}{n}$, the agents could consume from their top $j$ items using a sequential strategy. Consequently, each agent can guarantee an expected payoff of $\frac{1}{n}$, which implies the ratio between the optimal value and the value of the worst Nash equilibrium is at most $\frac{1}{4\sqrt{n}\cdot \log m}$.
    
Second, assume $OPT>4\sqrt{n}\log m$. Take a Nash equilibrium $\NEValue$. Now let $X_0$ be the set of items with consumption time between $0$ and $1$ and let $X_\ell$ be the set of items with consumption time between $2^{\ell-1}$ and $2^{\ell}$, for $\ell=1,\dots,\log (m/n)$. Denote by $OPT_\ell$ the contribution of items in $X_\ell$ to the value of $OPT$. 

We have two cases. Either $OPT_0\geq \frac{OPT}{2}$ or $\sum_{\ell=1}^{\log(m/n)} OPT_\ell\geq \frac{OPT}{2}$. In the former case, noting that $|X_0|\leq n$, the result follows from Theorem~\ref{thm:sym}. 
    In the latter case, there must exist some $\ell$ such that $OPT_{\ell}\geq \frac{OPT}{2\log m}>2\sqrt{n}$. We can then apply Theorem~\ref{thm:main-case-2}. Consider the set $\mathcal{A}$ of agents who have $2^\ell\cdot \sqrt{n}$ items. Because there are at most $2^\ell\cdot n$ items in $X_\ell$, if follows that $|\mathcal{A}|\le \sqrt{n}$. Given the valuations are unit-sum, this implies that the contribution of the remaining agents to the value of $OPT_\ell$ is at least $\frac{OPT}{4\log n}$. 

Next, for any $i\notin \mathcal{A}$, let the set $X_{i,\ell}$ consist of their favorite $2^{\ell-2}$ items (or their favorite item if $\ell=1$) from amongst the items whose consumption time is between $2^{\ell-1}$ and $2^{\ell}$ (or all items if there are less than $2^{\ell-2}$ such items). Then, $\sum_{j\in X_{i,\ell}} \TrueValue[i](j)$ is at least $\frac{1}{2\sqrt{n}}$ of the agent's contribution to $OPT_\ell$. In particular:
    \[\sum_{i\in [n]\setminus \mathcal{A}} \sum_{j\in X_{i,\ell}} \TrueValue[i](j) 
    \ \geq\  \frac{OPT}{4\log m}\cdot \frac{1}{2\sqrt{n}}
    \ =\ \frac{OPT}{8\sqrt{n}\cdot \log m}.\]

    That is, if each agent can be guaranteed a constant proportion of $X_{i,\ell}$ in the Nash equilibrium, then the desired result holds. Using the uniform strategy, we can now apply Theorem~\ref{thm:main-case-2} with a lower bound of $\frac{OPT}{16 \sqrt{n}\cdot \log m}$.
    Hence the agents can get a constant proportion of the $X_{i,\ell}$ and their contribution to the optimal welfare is $\fn{\Omega}{\frac{OPT}{\sqrt{n}\log m}}$. So the price of anarchy is at most $O(\sqrt{n}\cdot \log m)$.
    \qed
\end{proof}

We remark that this proof is for pure Nash equilibria. However, as we show in Appendix~\ref{sec:mixed}, the result extends to mixed Nash equilibria and to coarse correlated equilibria. Furthermore, mixed Nash equilibria and coarse correlated equilibria are guaranteed to exist in this model.

\subsection{An Upper Bound on the Price of Anarchy for Probabilistic Serial}\label{sec:UB-PS}

We now use the ideas we have developed for Cardinal Probabilistic Serial to obtain the same upper bound on the price of anarchy for Probabilistic Serial. 
Fortunately, both the single-minded strategies and the sequential bidding strategies are applicable to Probabilistic Serial. 
Moreover, since we have bounds on the value obtained by each agent by switching to the sequential strategy with Cardinal Probabilistic Serial, 
the same bound also applies to Probabilistic Serial. 
Indeed, the set of possible strategies for the remaining agents for Probabilistic Serial is contained within the set of strategies for Cardinal Probabilistic Serial.
This will imply that the bounds obtained using only the sequential strategy still hold in the special case where $m\leq n$. 

Unfortunately, for the case $m>n$, we cannot extend the proof idea for items with a large completion time to work for Probabilistic Serial.
In particular, the uniform bidding strategy does not work for Probabilistic Serial which is an ordinal mechanism.
Instead, we apply a different approach based upon the following lemma.

\begin{lemma}\label{lm:UB-PS}
    For Probabilistic Serial, for any agent $i$, let $v'_i$ be the true value $i$ has for the items and let $v$ be any pure Nash equilibrium with respect to $v'$. 
    Then, for any sequence of items $X=\set{x_1,x_2,\dots,x_k}$ whose consumption times are bounded below by $q\geq 2$, for some $k< \floor{q}$ it holds that: 
    \[v'_i(CPS(v))\ \geq\ \sum_{\ell=1}^k \TrueValue[i](x_{\ell}).\]
\end{lemma}

\begin{proof}
    If agent $i$ bids for the items of $X$ from time $0$ to $\floor{q}-1$, then the remaining items will be consumed more slowly. In particular, it will take longer for the remaining agents to switch from consuming the other items to consuming items from $X$. But the consumption time of items in $X$ is at least $q$ and the minimum non-zero consumption rate in Probabilistic Serial is 1.
    Therefore, no agent may consume from the set $X$ before time $\floor{q}-1$. Thus, $i$ is the only agent consuming these items so it consumes all of $X$ for a duration of time $k$. Hence, $i$ is guaranteed to win all the items in $X$ and the result follows.
\end{proof}

We can now prove our upper bound for the price of anarchy of Probabilistic Serial. 

\PSAnarchy*

\begin{proof}
First, assume that $OPT\leq 4\sqrt{n}\cdot \log m$. Then a truthful report of the preferences for Probabilistic Serial will guarantee each agent an expected payoff of $\frac{1}{n}$ (similarly to the sequential strategy for Cardinal Probabilistic Serial).
    This implies that the ratio between the optimal value and the value of the worst Nash equilibrium is at most $\frac{1}{4\sqrt{n}\log m}$.
    
Second, assume $OPT>4\sqrt{n}\log m$. As in the proof of Theorem~\ref{thm:CPS-Anarchy}, denote the contribution of the items whose consumption time is at most $1$ by $OPT_0$ and the contribution of those item with consumption time between $2^{\ell-1}$ and $2^{\ell}$ as $OPT_\ell$.
    We have two cases.
   If $OPT_0\geq \frac{OPT}{2}$ then the same proof we used for Cardinal Probabilistic Serial works because using the sequential strategy for our bound is still valid. 

   If $\sum_{\ell=1}^{\log(m/n)} OPT_\ell\geq \frac{OPT}{2}$ then take $\ell$ maximizing $OPT_\ell$. We can now find disjoint sets $X_{i,\ell}$ so that $|X_{i,\ell}|\leq 2^{\ell-1}$ and $\sum_{i=1}^n \sum_{j\in X_{i,\ell}}\TrueValue[i](j)\geq \frac{OPT}{8\sqrt{n}\log m}$.
   If each agent can be guaranteed a constant proportion of $X_{i,\ell}$ in the Nash equilibrium, then the result holds. Using Lemma~\ref{lm:UB-PS}, if $\ell\neq 1$, agent $i$ can obtain all of $X_{i,\ell}$ giving a bound of $\frac{OPT}{8\sqrt{n} \log m}$. On the other hand, if $\ell=1$, the argument used for Cardinal Probabilistic Serial can be applied to show the agent gets a half unit of each item in $X_{i,\ell}$. This gives a lower bound of $\frac{OPT}{16\sqrt{n}\log m}$.
     Hence the agents can get a constant proportion of the $X_{i,\ell}$ and their contribution to the optimal welfare is $\fn{\Omega}{\frac{OPT}{\sqrt{n}\log m}}$. So the price of anarchy is at most $O(\sqrt{n}\cdot \log m)$.
\end{proof}

Ergo, the price of anarchy bound for the asymmetric one-side allocation problem applies to both
Probabilistic Serial and Cardinal Probabilistic Serial.

\subsection{A Lower Bound on the Price of Anarchy}\label{sec:LB} 
We now present two lower bounds on the price of anarchy.
For our first lower bound, we verify that Theorem~\ref{Anarchy-Gen} extends to the symmetric one-sided allocation problem with multi-unit demand agents.
\begin{theorem}\label{Anarchy-Alloc}
For the symmetric one-side allocation problem, the pure price of anarchy of any unit-sum fair mechanism is $\Omega(\sqrt{n})$.
\end{theorem}
\begin{proof}
Consider the example used by Christodoulou et al.~\cite{CFF15} to prove Theorem~\ref{Anarchy-Gen} for the matching problem.
Take the following valuation function:
\[v_i(j)=\begin{cases}
\frac{1}{n}+\varepsilon\text{ if }i=j\cdot \sqrt{n}+i'\text{ for }i'=1,\dots,\sqrt{n}\\
\frac{1}{n}-\frac{\varepsilon}{n-1}\text{ otherwise}
\end{cases}\]
Now consider a Nash equilibrium for $v$.  
Let $i_j$ be the index of the agent who has positive value for item $j$
but has the smallest probability of being assigned $j$ in the Nash equilibrium.
Next, create a new valuation $v'_i(j)$ which is $v_i(j)$ if $i\neq i_{j'}$ for any $j'$ and which is $1\text{ if }i=i_j$ and $0\text{ if }i=i_{j'}\neq i_j$

Since the agents get the same number of items in expectation, a Nash equilibrium for $v$ is also a Nash equilibrium for $v'$ 
where the agents maximize their probability of getting their favorite item.
The social welfare of the optimal allocation is $\sqrt{n}$. 
At the Nash equilibrium, since the agents $i_j$ get assigned $j$ with probability at most $\frac{1}{\sqrt{n}}$,
the social welfare is at most $\sqrt{n}\cdot \frac{1}{\sqrt{n}}+\sqrt{n}\cdot \left(1-\frac{1}{\sqrt{n}}\right)\cdot \left(\frac{1}{n}+\frac{1}{n^3}\right)\leq 3$.
This gives a lower bound of $\Omega(\sqrt{n})$ on the price of anarchy.
\qed
\end{proof}

This bound is not surprising as with multi-unit demand agents the optimal allocation has higher welfare than the optimal matching.
However, our second lower bound is more surprising: for the asymmetric setting, the price of anarchy deteriorates with the number of items!

\CPSAnarchyAsym*

\begin{proof}
Assume we are given $k$ and let $q=o(k)$.
Let the number of agents be $n=k+q$ and the number of items be $m=2^{q}-1$.
Assume that, for $i=1,\dots,k$, agent $i$ has value $1$ for item $1$ and $0$ for the remaining items.
For $i=k+1,\dots,k+q$, agent $i$ has value $\frac{1}{2^i}$ for items $2^{i-k}$ to $2^{i+1-k}-1$. 
Assume, without loss of generality, that when the items they have positive value for have been consumed, 
agents consume from the lowest indexed item at a rate of 1.\footnote{This can be done without loss of generality using the justification for the sequential strategy given in Appendix~\ref{sec:epsilon}.}

Then, the items will be consumed in order and until item $j$ has been consumed, agents $1$ to $k+\floor{\log j}$ are only consuming items $1$ to $j$. So item $j$ is consumed at time $\frac{j}{k+\floor{\log j}}\geq \frac{2j}{n}$.
Thus, agents $1$ to $k$ are consuming item $1$ which they value at $1$ for a time duration of $\frac{1}{k}<\frac{4}{n}$.
Then agents $k+1$ to $k+q$ will consume items $2^{i-k}$ to $2^{i+1-k}-1$, which they value at $\frac{1}{2^{i-k}}$, at a rate of $1$ for a time duration of at most $\frac{2^{i+1-k}-1}{n}$. Hence the value they obtain is at most $\frac{2^{i+1-k}-1}{n}\cdot \frac{1}{2^{i-k}}<\frac{4}{n}$.

Since every agent has value at most $\frac{4}{n}$ from the mechanism, the total social welfare is at most $4$. However, the optimal allocation has welfare $q$ so the price of anarchy of Probabilistic Serial and Cardinal Probabilistic Serial is at least $\Omega(q) =\Omega(\log m)$.

Note, it is easy to see that $O(n)$ is always an upper bound on the price of anarchy in any instance. 
\qed
\end{proof}

We conjecture that the price of anarchy for $n$ agents and $m$ items is, in fact, $\Theta(\min\set{n\, ,\, \sqrt{n}\cdot \log m})$. This would imply our
upper bound of $O(\sqrt{n}\cdot \log m)$ is tight.

\section{Related Problems}\label{sec:related}
\subsection{The Price of Stability}\label{sec:price-of-stability}
We obtain similar bounds for the price of stability.
Below the upper bound follows immediately from our price of anarchy bound. The lower bound is given in Appendix~\ref{app:price-of-stability}.
\begin{theorem}\label{thm:CPS-Stability}
For the one-sided allocation problem with multi-unit demand agents, the price of stability 
of Cardinal Probabilistic Serial is at least $\Omega(\sqrt{n})$ and at most $O(\sqrt{n}\cdot \log n)$.
\end{theorem}

\subsection{The Relative Merits of Probabilistic Serial and Cardinal Probabilistic Serial}\label{sec:compare}

The reader may ask how Probabilistic Serial and Cardinal Probabilistic Serial compare in {\em practical} performance.
Interestingly, they are not directly comparable. In Appendix~~\ref{sec:relative}
we present an instance (Example~3) where CPS dramatically outperforms PS, 
but also an instance (Example~4) where PS dramatically outperforms CPS, under truthful reporting.
As a rule of thumb,  CPS performs much better when 
there is a consensus among the agents on the relative ranking of the items,
but where agents may have very different specific valuations for the items;
estate division and sports drafts are settings where this characteristic naturally arises.
On the other hand, if the agents are primarily interested in disjoint sets
of items then PS performs better as then agents greedily consume those items as quickly as possible.
Investigating the practical performances of PS and CPS via
comprehensive experimental analyses is an important future direction.

\subsection{The Random Priority Mechanism}\label{sec:RSD}
The focus of this paper has been Probabilistic Serial mechanisms (ordinal and cardinal). How does the other classical mechanism, Random Priority, perform in the asymmetric one-sided allocation problem?
To answer this question, we remark that there are two natural ways to implement Random Priority in the asymmetric setting:
\begin{itemize}
    \item {\em Random Priority}: 
    Each agent is randomly sampled once and, upon selection, picks their favorite $\frac{m}{n}$ items from amongst those that are still available.
    \item {\em Repeated Random Priority}: 
 Agents are sampled repeatedly $m$ times in a row (uniformly and independently) and, upon selection, the selected agent picks its favorite available item.
\end{itemize}
\begin{theorem}
  For the asymmetric one-sided allocation problem with multi-unit demand agents, the price of anarchy 
    of Random Priority is at least $n$ when $m\geq n^2$.
\end{theorem}
\begin{proof}
    The upper bound follows because each agent is guaranteed a payoff of at least $\frac{1}{n}$.

    The lower bound follows from assuming we have $m=n^2$ items and agent $i$ has value $1-\varepsilon$ for item $i$ and $\frac{\varepsilon}{n-1}$ for items $i'\in [n]\setminus \{i\}$ (and 0 for the rest). Then, the first agent to be selected gets value 1 but the remaining agents get value 0. In the optimal allocation, every agent gets value $1-\varepsilon$, so the optimal welfare is $n\cdot(1-\varepsilon)$. 
    \qed
\end{proof}

\begin{theorem}
    For the asymmetric one-sided allocation problem with multi-unit demand agents, the price of anarchy 
    of Repeated Random Priority is at least $\Omega(\min \{n\, ,\, \log m\})$.
\end{theorem}
\begin{proof}
    Consider the example from the proof of Theorem~\ref{thm:CPS-Anarchy-m}. The expected value for an agent who get $2^i$ items is $\frac{1}{2^i}$ multiplied by their expected number of items won from the set of items they are interested in. This is at most the number of times that appear in the first $2^{i+1}$ rounds, omitting those where an agent with a higher index wins. This can be bounded by the expectation of a binomial $B\left(2^{i+1},\frac{1}{k+i}\right)$ which is $\frac{2^{i+1}}{k+i}\leq \frac{2^{i+2}}{n}$. This implies that their expected value is at most $\frac{4}{n}$. Thus the welfare of the allocation is at most $4$, even though the optimal social welfare is $q$. Consequently, when there are $2^q-1$ items the welfare of the best Nash equilibrium is at most $\frac{OPT}{q}=\frac{OPT}{\log m}$.
    \qed
\end{proof}

\section{Conclusion}\label{sec:conclusion}
We studied fair mechanisms for the asymmetric one-sided allocation problem with multi-unit demand agents. A natural open problem is to close the logarithmic gap between the upper and lower bounds in 
Theorem~\ref{thm:CPS-Anarchy} and Theorem~\ref{Anarchy-Alloc}.
Another interesting line of research is to study whether our results  extend to other classes of valuation function, specifically, 
non-additive valuation functions. We remark that while the price of anarchy bounds for the unit demand setting extend to {\em unit-range} valuations this is {\em not} the case for multi-unit demands. 
For unit-range valuations the price of anarchy is $\Omega(n)$ for any fair mechanism.\footnote{To see this, take a single agent with value 1 for every item and let other agents having value $1$ for the first item and $\varepsilon$ for the remaining items.} 

\bibliographystyle{plain}
\bibliography{references}

\appendix

\section{Epsilon Strategies and Sequential Bidding}\label{sec:epsilon}

In this section we show that an agent can mimic the sequential bidding strategy with arbitrary precision using epsilon-valuation strategies.
Recall, given a sequence $X$ of length $k$,
the epsilon-strategy $\EpsValue{X}\in\valInd$ is defined by
$$
\EpsValue{X} (j) =
\begin{cases}
1-\sum_{\ell=1}^{k-1} \varepsilon^{\ell} &\text{ if }j=x_1\\\varepsilon^\ell    &\text{ if }j=x_\ell\\
0 &\text{ otherwise}
\end{cases}
$$
The limit of the epsilon-strategy $\EpsValue{X}$ when $\varepsilon\rightarrow 0$ is the sequential strategy $\SeqValue{X}$. 
Here we will formally justify allowing the sequential strategy in the mechanism.
\begin{lemma}\label{lem:epsilon}
$\forall i\in [n]$, $\forall v\in \valGroup$, $\forall \delta>0$, $\exists \varepsilon>0$ 
such that $CPS(\TrueValue[i],(\EpsValue{X},v_{-i}))\geq (1-\varepsilon)\cdot CPS(\TrueValue[i],(\SeqValue{X},v_{-i}))$.
That is, the payoff of the epsilon-strategy is within 
$\delta$ of the payoff of the sequential strategy.
\end{lemma}

\begin{proof}
For convenience, in the proof we assume that the order of the items in the sequential strategy is the same as the completion time. 
If that is not the case for some item $j$, then compared to the case where $i$ is using an epsilon-strategy, the difference between the amount consumed in both cases is at most $\varepsilon$ while the difference between the amount consumed by the remaining agents is bounded by the difference between the consumption times of the items which precede, which is bounded by the rest of our proof.

Let $\Delta_j = \NETime{j} - \NETime{j-1}$ with $\Delta_1 = t_1$ in the sequential strategy. 
Let $\Delta'_j = \tau_{j} - \tau_{j-1}$ with $\Delta_1 = \tau_1$ where $\tau_j$ is the consumption time of $j$ in the epsilon-strategy. 
Let $i'$ be any agent and $i$ be the agent that changes strategy.

Then we have that for the sequential strategy:
\[\Delta_j = \frac{1 - \sum_{j'=1}^{j-1}\Delta_{j'} \cdot (\sum_{i'=1, i' \neq i}^{n} \frac{v_{i'}(j)}{1-\sum_{\tilde{j}=1}^{j'}v_{i'}(\tilde{j})})} {1+\sum_{i' \neq i} \frac{v_{i'}(j)}{1-\sum_{\tilde{j}=1}^{j-1}v_{i'}(\tilde{j})}}.\]

For the epsilon-strategy:
\begin{align*}
    &\Delta'_j\\ = &\frac{1 - \sum_{j'=1}^{j-1}\Delta'_{j'} \cdot (\sum_{i'=1, i' \neq i}^{n} \frac{v_{i'}(j)}{1-\sum_{\tilde{j}=1}^{j'-1}v_{i'}(\tilde{j})})-\sum_{j'=1}^{j-1}\Delta'_{j'}\cdot (\varepsilon^{j-j'}f_j)} {f_j + \sum_{i' \neq i} \frac{v_{i'}(j)}{1-\sum_{\tilde{j}=1}^{j-1}v_{i'}(\tilde{j})}},
\end{align*}
where \[f_j = \frac{\varepsilon^j}{1-1+\sum_{j'=1}^{k}\varepsilon^{j'}-\sum_{j'=1}^{j-1}\varepsilon^{j'}} = \frac{1}{\sum_{j'=0}^{k-j}\varepsilon^{j'}}=1-\frac{\sum_{j'=1}^{k-j}\varepsilon^{j'}}{\sum_{j'=0}^{k-j}\varepsilon^{j'}}.\]

Now let $A^{(j)}_{j'}$ be the consumption rate of $j$ by agents aside from $i$ when the items up to $j'$ have been consumed.
That is:
\[A^{(j)}_{j'} = \sum_{i'=1, i' \neq i}^{n} \frac{v_{i'}(j)}{1-\sum_{\tilde{j}=1}^{j'}v_{i}(\tilde{j})}\]

Then we can simplify the expressions for $\Delta_j$ and $\Delta'_j$ to 
\begin{align*}\Delta_j &= \frac{1 - \sum_{j'=1}^{j-1}\Delta_{j'} \cdot A^{(j)}_{j'}} {1+A^{(j)}_{j}}\\
\Delta'_j &= \frac{1 - \sum_{j'=1}^{j-1}\Delta'_{j'} \cdot (A^{(j)}_{j'} + \varepsilon^{j-j'}\cdot f_{j'})} {f_j + A^{(j)}_{j}}
\end{align*}

This gives us the following equation:
\begin{align*}
\Delta_j\ =\ &\frac{1 - \sum_{j'=1}^{j-1}\Delta_{j'} \cdot A^{(j)}_{j'}} {1+A^{(j)}_{j}}
\\
=\ &\frac{(f_j+A^{(j)}_{j})\cdot(1 - \sum_{j'=1}^{j-1}\Delta_{j'} \cdot A^{(j)}_{j'})} {(1+A^{(j)}_{j})\cdot(f_j+A_j)}\\
=\ &\frac{(f_j-1)\cdot(1 - \sum_{j'=1}^{j-1}\Delta_{j'} \cdot A^{(j)}_{j'})} {(1+A^{(j)}_{j})\cdot(f_j+A_j)}
\\&\qquad+\frac{(1+A^{(j)}_{j})\cdot(1 - \sum_{j'=1}^{j-1}\Delta_{j'} \cdot A^{(j)}_{j'})} {(1+A^{(j)}_{j})\cdot(f_j+A_j)}\\
=\ &\frac{(f_j-1)\cdot(1 - \sum_{j'=1}^{j-1}\Delta_{j'} \cdot A^{(j)}_{j'})} {(1+A^{(j)}_{j})\cdot(f_j+A_j)}
+\frac{1 - \sum_{j'=1}^{j-1}\Delta_{j'} \cdot A^{(j)}_{j'}} {f_j+A_j}
\end{align*}

So, when taking $\Delta_j-\Delta'_j$ we get the following:

\begin{align*}
\Delta_j-\Delta'_j
\ =\ &\frac{(f_j-1)\cdot(1 - \sum_{j'=1}^{j-1}\Delta_{j'} \cdot A^{(j)}_{j'})} {(1+A^{(j)}_{j})\cdot(f_j+A_j)}
+\frac{1 - \sum_{j'=1}^{j-1}\Delta_{j'} \cdot A^{(j)}_{j'}} {f_j+A_j}
\\
&\qquad -\frac{1 - \sum_{j'=1}^{j-1}\Delta'_{j'} \cdot (A^{(j)}_{j'} + \varepsilon^{j-j'}\cdot f_{j'})} {f_j + A^{(j)}_{j}}\\
=\ &\frac{(f_j-1)\cdot(1 - \sum_{j'=1}^{j-1}\Delta_{j'} \cdot A^{(j)}_{j'})} {(1+A^{(j)}_{j})\cdot(f_j+A_j)}
\\
&\qquad +\frac{\sum_{j'=1}^{j-1}\Delta'_{j'} \cdot (A^{(j)}_{j'} + \varepsilon^{j-j'}\cdot f_{j'})- \Delta_{j'} \cdot A^{(j)}_{j'}} {f_j + A^{(j)}_{j}}
\end{align*}

Remark that $1\geq f_j\geq 1-\varepsilon$ and that $\sum_{j'=1}^{j} \Delta'_j<1$
So, by taking the absolute value, we get:

\begin{align*}
|\Delta_j-\Delta'_j|\ \leq \ &\abs{\frac{(f_j-1)\cdot(1 - \sum_{j'=1}^{j-1}\Delta_{j'} \cdot A^{(j)}_{j'})} {(1+A^{(j)}_{j})\cdot(f_j+A_j)}}
\\
&\qquad +\abs{\frac{\sum_{j'=1}^{j-1}\Delta'_{j'} \cdot (A^{(j)}_{j'} + \varepsilon^{j-j'}\cdot f_{j'})- \Delta_{j'} \cdot A^{(j)}_{j'}} {f_j + A^{(j)}_{j}}}\\
\leq\ &\frac{\abs{f_j-1}}{(1+A^{(j)}_{j})\cdot(f_j+A^{(j)}_{j})}+\frac{\left |\sum_{j'=1}^{j-1}(\Delta'_{j'}-\Delta_{j'})\cdot A^{(j)}_{j'} \right |} {f_j+A^{(j)}_{j}}\\
&\qquad +\abs{\frac{\sum_{j'=1}^{j-1}\Delta'_{j'} \cdot  \varepsilon^{j-j'}\cdot f_{j'}} {f_j + A^{(j)}_{j}}}\\
\leq\ &\frac{\varepsilon}{\frac{1}{2}}+\abs{|\sum_{j'=1}^{j-1}(\Delta'_{j'}-\Delta_{j'})}\cdot \frac{A^{(j)}_{j'}} {f_j+A^{(j)}_{j}}+\varepsilon\cdot\frac{f_{j'}}{f_j + A^{(j)}_{j}}\\
\leq\ &3\varepsilon+\sum_{j'=1}^{j-1}\abs{\Delta'_{j'}-\Delta_{j'}}
\end{align*}

So we get $|\Delta_j-\Delta'_j|\leq 3j\varepsilon$.

So, the set of remaining items only changes for a time of at most $\sum_{j'=1}^j |\Delta_{j'}-\Delta'_{j'}|\leq 3j^2\varepsilon$ which implies that the payoff for the agents $i'\neq i$ changes by at most $3j^2\varepsilon$ since they have unit-sum valuations. 

On the other hand, for $i$ when both mechanisms agree on the set of remaining items $i$ only changes the item they are consuming by less than $2\varepsilon$, in particular, if we sum the difference between what is consumed when the mechanisms disagree and when they agree $i$'s consumption only changes by at most $2\varepsilon+3j^2\varepsilon\leq 4j^2\varepsilon$. Since $i$ has a unit-sum valuation, the change in the payoff is at most $4j^2\varepsilon$.

By setting $\varepsilon=\frac{\delta}{4j^2}$ we get the result we wanted.
\end{proof}

\section{Mixed Strategies}\label{sec:mixed}
Here we show that our main result, the upper bound on the price of anarchy for pure strategy Nash equilibria in symmetric instances, also applies to mixed strategy Nash equilibria and coarse correlated equilibria.
To show this we use the following definitions and notations.
A mixed strategy for an agent $i$ is a probability distribution of $\valInd$ and is denoted as $p_i:\valInd \rightarrow [0,1]$. The mixed strategy used by every agent is denoted as $p\colon \valGroup\rightarrow [0,1]$ with $p(v)=\prod_{i=1}^n p_i (v_i)$. 

\begin{theorem}\label{Mixed CPS-Anarchy}
The price of anarchy of coarse correlated equilibria is $O(\sqrt{n}\cdot \log n)$ in the one-sided allocation problem with multi-unit demand agents.
\end{theorem}
\begin{proof}
Recall Lemma~\ref{lem:lower-bound} 
states that for any pure Nash equilibrium $v$ and for any sequence of items $X=\set{x_1,x_2,\dots,x_k}$:
\[\TrueValue[i](CPS(v))\geq \frac{1}{4} \sum_{\ell=1}^k (\NETime{x_{\ell}}-\NETime{x_{\ell-1}})\cdot \TrueValue[i](x_{\ell}).\]

To prove this, we bounded the payoff $i$ obtained by deviating to the sequential bidding strategy. This also applies for mixed strategies. 
In particular, if the $x_\ell$ are ordered by $i$'s value for them, 
then by deviating to the sequential strategy from any pure strategy
agent $i$ can consume item $x_\ell$ from time $\frac{1}{4}\max_{\ell'=1,\dots,\ell-1} \NETime{x_{\ell'}}$
to time $\frac{1}{4}\max_{\ell'=1,\dots,\ell} \NETime{x_{\ell'}}$.
By the linearity of the expectation, this gives the following bound for a mixed strategy:
\begin{eqnarray*}
\lefteqn{
\fn{\bb{E}}{\sum_{\ell=1}^{k_i} \TrueValue[i](x_{\ell}) \cdot \left(\frac{1}{4}\max_{\ell'=1,\dots,\ell} \NETime{x_{\ell'}}-\frac{1}{4}\max_{\ell'=1,\dots,\ell-1} \NETime{x_{\ell'}}\right)}}\\
&\geq& 
\frac{1}{4}\cdot \sum_{\ell=1}^{k_i} \TrueValue[i](x_{\ell}) \cdot \left(\fn{\bb{E}}{\max_{\ell'=1,\dots,\ell} \NETime{x_{\ell'}}}-\fn{\bb{E}}{\max_{\ell'=1,\dots,\ell-1} \NETime{x_{\ell'}}}\right)
\end{eqnarray*}
We can now use this bound and apply the same proof as in Theorem~\ref{thm:CPS-Anarchy}
to obtain the same upper bound on the price of anarchy for mixed equilibria. A similar argument applies for coarse correlated equilibria.
\end{proof}
We remark that mixed Nash equilibria and coarse correlated equilibria
are guaranteed to exist.

\section{Lower Bounds}
\subsection{Tightness of Proof Methodology}\label{app:tight}
We show here that the tools utilized in this paper are not strong enough to remove the logarthmic term in Theorem~\ref{thm:sym}.
In fact, we conjecture that $\Theta(\sqrt{n}\cdot \log n)$ is a tight bound for the symmetric case.
In particular, the bound from 
Lemma~\ref{lem:time}
is too loose and so will 
induce a logarithmic term in the upper bound. 
Namely, assuming that the items are consumed in increasing order, 
then substituting $t_j$ by $j/n$ will lead to the appearance of a $\log$ factor.

\begin{lemma}
$O(\sqrt{n}\cdot \log^{O(1)} n)$ is a tight bound when bounding $\NETime{j}$ below by $j/n$.
\end{lemma}
\begin{proof}
Consider the following example.
There are $x\cdot k$ agents. For each $z=0,\dots,x-1$,  
there are exactly $k$ agents who are assigned $2^z$ items in the optimal allocation and have value $\frac{1}{2^z}$ for each item. 
Hence there are $n=(2^x-1)\cdot k$ items. Setting $k=(2^x-1)/x^2$, we have $n=(xk)^2=((2^x-1)^2)/x^2$.

Note that each agent will individually consume the items they are meant to be assigned at a rate of at most $1/n$ so unless other agents consume these the consumption time will be 1.
In particular, this implies that using the remaining $n-z\cdot k$ agents, we can choose the order in which the items are consumed.
So, assume that the items of higher value, that is those assigned to agents with smaller bundles, are consumed faster.

Let $\mathcal{I}_z$ be the set of agents who receive $2^z$ items in the optimal allocation. 
Let $\mathcal{J}_z$ be the set of items assigned to agents in $\mathcal{I}_z$. 
Let $prec(j)$ be the set of items that have been consumed before or at the same time as $j$ (including $j$).
Then, for any $j\in\mathcal{J}_z$, an upper bound on the number of items that have been consumed before $j$, is the number of items that are assigned to agents with at most $2^z$ items. That is:
\[|prec(j)|\leq \abs{\bigcup_{z'=0}^z \mathcal{J}_{z'}}=\sum_{z'=0}^z 2^{z'}\cdot k=(2^{z+1}-1)\cdot k\leq 2^z\cdot k\]

In particular, by denoting $X_i=\set{x_1^{i},\dots,x_{2^z}^{i}}$ to be the set of items $i\in \mathcal{I}_z$ gets, the bound we get for the value of the allocation when substituting the time by $prec(j)/n$ is the following:

\begin{align*}
    \sum_{z=0}^x \sum_{i\in \mathcal{I}_z} \frac{1}{2^z}\cdot \frac{\sup_{\ell=1,\dots,2^z} prec(x_{\ell}^{(i)})}{n}\ \leq\ \sum_{z=0}^{x}\left|\mathcal{I}_z\right|\cdot \frac{1}{2^z}\cdot \frac{2^{z+1}\cdot k}{n}\\
    \ =\ \sum_{z=0}^{x} \frac{2k^2}{n} 
    \ =\ \frac{2xk^2}{(xk)^2} 
    \ =\ \frac{2}{x}
\end{align*}
This means that our bound will only prove $O(x\sqrt{n})$ and given that $n=\frac{(2^x-1)^2}{x^2}$, we get that $x=\log^{O(1)}(n)$.
\end{proof}

\subsection{Lower Bound of the Price of Anarchy}
Recall that by Theorem~\ref{Anarchy-Gen}, for unit demand, for the one-sided matching problem with unit sum-valuation has a price of anarchy of $\Omega(\sqrt{n})$ for any mechanism. 

We verify that the theorem extends to the one-side allocation problem with multi-unit demand agents.

\begin{proof}[Proof Sketch]
Consider the example used by Christodoulou et al.~\cite{CFF15} to prove Theorem~\ref{Anarchy-Gen} for the matching problem.
Take the following valuation function:
\[v_i(j)=\begin{cases}
\frac{1}{n}+\varepsilon\text{ if }i=j\cdot \sqrt{n}+i'\text{ for }i'=1,\dots,\sqrt{n}\\
\frac{1}{n}-\frac{\varepsilon}{n-1}\text{ otherwise}
\end{cases}\]
Now consider a Nash equilibrium for $v$.  
Let $i_j$ be the index of the agent who has positive value for item $j$
but has the smallest probability of being assigned $j$ in the Nash Equilibrium.
Next, create a new valuation $v'_i(j)$ which is $v_i(j)$ if $i\neq i_{j'}$ for any $j'$ and which is $1\text{ if }i=i_j$ and $0\text{ if }i=i_{j'}\neq i_j$

Since the agents get the same number of items in expectation, a Nash equilibrium for $v$ is also a Nash equilibrium for $v'$ 
where the agents maximize their probability of getting their favorite item.
The social welfare of the optimal allocation is $\sqrt{n}$. 
At the Nash equilibrium, since the agents $i_j$ get assigned $j$ with probability at most $\frac{1}{\sqrt{n}}$,
the social welfare is at most $\sqrt{n}\cdot \frac{1}{\sqrt{n}}+\sqrt{n}\cdot \left(1-\frac{1}{\sqrt{n}}\right)\cdot \left(\frac{1}{n}+\frac{1}{n^3}\right)\leq 3$.
This gives a lower bound of $\Omega(\sqrt{n})$ on the price of anarchy.
\end{proof}

\section{The Price of Stability}\label{app:price-of-stability}

Here we study the price of stability.
To do this, we say that a strategy $u\in \valInd$ is a {\em safety strategy} for agent $i$ if $\forall v\in\valGroup$ the allocation output on input $(u,v_{-i})$ 
gives $i$ $\frac{k}{n}$ of its top $k$ items in expectation. For the one-sided matching problem under Random Priority and Probabilistic Serial, truthtelling is known to be a safety strategy.

Similar to the price of anarchy, the {\em price of stability} is the worst case ratio between the optimal welfare
and the social welfare of the best Nash equilibrium, namely:
$$
\sup_{v'}\, \inf_{v\in NE(v')}\, \frac{OPT(v')}{\sum_{i\in I} v'_i(M(v))}
$$

Interestingly, the existence of safety strategies induces the following bound on the price of stability for the one-sided matching problem:
\begin{theorem}[\cite{CFF15}]\label{Stability-Gen}
For the one-sided matching problem, the pure price of stability of any mechanism with a safety strategy is $\Omega(\sqrt{n})$.
\end{theorem}

As we did for the lower bound on the price of anarchy of general mechanisms, 
first we show that this lower bound extends to our setting.

\begin{theorem}\label{Stability-Alloc}
For the one-sided allocation problem, the pure price of stability of any mechanism with a safety strategy is $\Omega(\sqrt{n})$.
\end{theorem}
\begin{proof}
The example used by Christodoulou et al.~\cite{CFF15} to prove Theorem~\ref{Stability-Gen} for matchings suffices.
Consider the following valuation:
$$v_i(j)=
\begin{cases}
1 &\text{if }i=j\leq \sqrt{n}\\
\frac{1}{\sqrt (n)} &\text{if }i>\sqrt{n}\geq j\\
0&\text{otherwise}
\end{cases}
$$

Then clearly the optimal allocation is to assign an item to the agent who has value 1 for it, if possible, and to assign the remaining items in any way. Denoting $p_{i,j}$ to be the probability assigning $j$ to $i$, then:
\[\sum_{i\in [n]} \sum_{j\in[\sqrt{n}]}p_{i,j}=\sqrt{n}.\]

However, since the mechanism has a safety strategy, the agents who are matched in the optimal solution get their top item with probability at least $1/n$ so, 
we get the following bound on the contribution of the remaining agents to the social welfare of the Nash equilibrium:
\[\sum_{i\in [n]\setminus [\sqrt{n}]}\sum_{j\in[\sqrt{n}]} p_{i,j}\cdot \frac{1}{\sqrt{n}}\leq \frac{1}{\sqrt{n}}\cdot \left(\sqrt{n}-\frac{1}{\sqrt{n}}\right)\leq 1\]

On the other hand, the agents who do not get matched can get their top $\sqrt{n}$ items
with probability at least $\frac{1}{\sqrt{n}}$, so we get the following bound on the contribution of the matched agents to the social welfare of the Nash Equilibrium:

\[\sum_{i\in[\sqrt{n}]}\sum_{j\in[\sqrt{n}]} p_{i,j}\leq \sqrt{n}-(n-\sqrt{n})\cdot \frac{1}{\sqrt{n}}=1\]

So, the contribution of all the agents to the social welfare of the Nash equilibrium
is at most 2. But the optimal allocation clearly has value $\sqrt{n}$.
\end{proof}

So, if we can show that Cardinal Probabilistic Serial has a safety strategy then we get a lower bound of $\Omega(\sqrt{n})$ on the price of stability. However, interestingly, unlike it's ordinal counterpart, 
truthtelling is not a safety strategy.

\begin{lemma}
Truthtelling is \textbf{NOT} a safety strategy for Cardinal Probabilistic Serial.
\end{lemma}
\begin{proof}
Assume that $v_{i}(j)=1$, $v_{i}(j)=0$, $v_{1}(1)=1-(n-1)\varepsilon$ and $v_{1}(j)=\varepsilon$ for any $i\neq 1$ and $j\neq 1$. 
Then if agent $1$ is truthful it has a probability less than $1/n$ of getting item $1$, which is its top item. So, truthtelling is not a safety strategy.
\end{proof}
Nonetheless, we can find a safety strategy for Cardinal Probabilistic Serial.
\begin{lemma}\label{lem:Safe}
Cardinal Probabilistic Serial has a safety strategy.
\end{lemma}
\begin{proof}
This follows directly from Lemma~\ref{lem:time}
by considering the sequential strategy with $X_i=[n]$. Before time $j/n$, at most $j-1$ items have been consumed so under the sequential strategy agent $i$ is consuming from one of their $j$ favorite items. 
This is what we need.
\end{proof}

\begin{corollary}
For the one-sided allocation problem, the price of stability of Cardinal Probabilistic Serial is $\Omega(\sqrt{n})$ and $\bigO{\sqrt{n}\cdot \log~n}$
\end{corollary}
\begin{proof}
CPS has a safety strategy by Corollary~\ref{lem:Safe}. So, the lower bound follows Theorem~\ref{Stability-Alloc} which states that any mechanism with a safety strategy has a price of stability of $\Omega(\sqrt{n})$.
The upper bound follows from Theorem~\ref{thm:CPS-Anarchy}
because the price of anarchy upper bounds the price of stability.

Given that the sequential bidding strategy can be used for PS as well, the statement applies to PS. 
\end{proof}

\section{The Relative Merits of PS and CPS}\label{sec:relative}

Here we give examples where PS and CPS have a major difference in performance. 
The reader may verify that in Example 3, CPS performs dramatically better than PS whereas in Example 4 PS performs dramatically better.

\begin{example}
Assume that $n$ is a square. Then, for $i=1,\dots,\sqrt{n}$, agent $i$ has value $1$ for item $i$ and 0 for other items. Then, assume that the remaining agents have value $\frac{1}{n}+\varepsilon$ for items $1$ to $\sqrt{n}$ and $\frac{1}{n}-\frac{\varepsilon}{\sqrt{n}-1}$ for the remaining items. Then, for $\varepsilon$ small enough the Optimal allocation and CPS have value $\Theta(\sqrt{n})$ while PS has value $\Theta(1)$.
\end{example}

\begin{example}
Assume that for $i=1,\dots,n$ agent $i$ has value $\frac{1}{\sqrt{n}}$ for item $i$ and value $\frac{1-\frac{1}{\sqrt{n}}}{n-1}$ for the remaining items. Then, the optimal allocation and PS will return a matching whose value is $\sqrt{n}$. On the other hand, the value of CPS for this instance will be $\Theta(1)$.
\end{example}

\end{document}